\documentclass[10pt]{article}

\usepackage{authblk}
\usepackage{caption}
\usepackage{amstext,amsmath,amssymb,amsfonts,bbm}
\usepackage[latin1]{inputenc}
\usepackage{epsfig}
\usepackage{dsfont}
\usepackage{hyperref}
\usepackage{amsthm}
\usepackage{subfigure}
\usepackage{color}
\usepackage{multirow}
\usepackage{psfrag}
\usepackage{graphicx}
\usepackage{extarrows}
\usepackage{braket}

\usepackage[lmargin=60pt,rmargin=60pt,tmargin=80pt,bmargin=80pt]{geometry}

\theoremstyle{plain}

\newtheorem{lemma}{Lemma}

\newtheorem{theorem}{Theorem}

\newtheorem{definition}{Definition}
\newtheorem{proposition}{Proposition}
\theoremstyle{definition}

\newcommand{\bP}{ {\bf P} } 
\newcommand{\bQ}{ {\bf Q} } 
\newcommand{\bA}{ {\bf A} } 
\newcommand{\ba}{ {\bf a} } 
\newcommand{\bS}{ {\bf S} }

\newcommand{\cM}{{\cal M}}
\newcommand{\cG}{{\cal G}}

\newcommand{\cB}{{\cal B}}
\newcommand{\cH}{{\cal H}}

\newcommand{\cS}{{\cal S}}

\newcommand{\be}{\begin{equation}}
\newcommand{\ee}{\end{equation}}

\allowdisplaybreaks[4]

\begin{document}

\title{\bf The $1/N$ expansion of the symmetric traceless and the antisymmetric tensor models in rank three}

\author[1]{Dario Benedetti}
\author[2]{Sylvain Carrozza}
\author[3]{Razvan Gurau}
\author[4]{Maciej Kolanowski}

\affil[1]{\normalsize\it Laboratoire de Physique Th\'eorique, CNRS, Univ.Paris-Sud, Universit\'e Paris-Saclay, 91405 Orsay, France \authorcr
email: dario.benedetti@th.u-psud.fr  \authorcr \hfill }

\affil[2]{\normalsize\it Perimeter Institute for Theoretical Physics
 31 Caroline St N, N2L 2Y5, Waterloo, ON, Canada  \authorcr email: scarrozza@perimeterinstitute.ca  \authorcr \hfill}

\affil[3]{\normalsize\it Centre de Physique Th\'eorique, CNRS, \'Ecole Polytechnique, Universit\'e Paris-Saclay, 91128 Palaiseau, France \authorcr
\it and Perimeter Institute for Theoretical Physics, 31 Caroline St. N, N2L 2Y5, Waterloo, ON, Canada \authorcr
email: rgurau@cpht.polytechnique.fr \authorcr \hfill}

\affil[4]{\normalsize\it Faculty of Physics, University of Warsaw, Pasteura 5, 02-093 Warsaw, Poland \authorcr
email: mp.kolanowski@student.uw.edu.pl \authorcr \hfill}

\date{}

\maketitle

\hrule\bigskip

\begin{abstract}
We prove rigorously that the symmetric traceless and the antisymmetric tensor models in rank three with tetrahedral interaction admit a $1/N$ expansion, and that at leading order they are dominated by melon diagrams.
This proves the recent conjecture of I. Klebanov and G. Tarnopolsky in \cite{Klebanov:2017nlk}, which they checked numerically up to 8th order in the coupling constant.
\end{abstract}

\bigskip

\hrule\bigskip

\tableofcontents

\bigskip

\section{Introduction and discussion}

Tensor models \cite{review,RTM} have recently come under the spotlight for their connection to the Sachdev--Ye--Kitaev (SYK)  
model \cite{Sachdev:1992fk,Kitaev,Polchinski:2016xgd,Jevicki:2016bwu,Maldacena:2016hyu,Gross:2016kjj,Gross:2017aos}, which is a promising testbed for better understanding black holes via an holographic description.
The two types of models share a similar large-$N$ limit \cite{Witten:2016iux,Klebanov:2016xxf}\footnote{However, they differ at subleading orders in $1/N$, see \cite{Bonzom:2017pqs}.}, but while in the SYK model this entails a large number of randomly distributed couplings, in tensor models it corresponds instead to a large symmetry group and just one or few couplings.
Therefore, tensor models offer an appealing alternative to the SYK model, and they are being thoroughly explored for such reason  
\cite{Peng:2016mxj,Krishnan:2016bvg,Krishnan:2017txw,Choudhury:2017tax,Giombi:2017dtl,Bulycheva:2017ilt,Krishnan:2017lra,Ferrari:2017jgw,Prakash:2017hwq,Peng:2017kro,Benedetti:2017fmp,BenGeloun:2017jbi}.

It is quite clear that the main feature of tensor models which makes them suitable for SYK-like applications is the existence of a large-$N$ limit \cite{color,expansion1}, and the fact that the class of diagrams that dominate in such limit are the melonic 
diagrams \cite{critical,uncoloring}.
This is a non-trivial result. In fact, while tensor models where introduced a long time ago \cite{ambj3dqg,sasa1}, in their original version they lacked a large-$N$ expansion and this severely limited their usefulness. Their large-$N$ behavior 
was dramatically improved by the addition of \emph{colors}. Originally introduced in the form of a coloring of the fields, \cite{color,expansion1}, it was later realized that the colors could be 
``pushed'' on the tensor indices \cite{uncoloring,sdequations,Carrozza:2015adg} to obtain models for only one tensor with \emph{no symmetry} under permutations: in other words, the location of a given 
index is an important label, and different indices transform independently under rotations.

Until recently, all the models having a well-defined large-$N$ limit were based on non symmetric tensors. 
On the other hand, while it has been understood since early on that completely symmetric tensors do not have a  well-defined large-$N$ limit, there are no no-go theorems for antisymmetric or traceless symmetric tensors, which are the subject of
the present paper. An important impulse to address the question of tensor with special symmetries came from the recent work of Klebanov and Tarnopolosky \cite{Klebanov:2017nlk}, where the authors consider models based either on a
symmetric traceless or on an antisymmetric tensor. They studied such models up to the 8th order in perturbation theory, and found no problematic diagrams. Therefore, they conjectured that such models admit a large-$N$ limit dominated by melonic 
diagrams.

Recently the $1/N$ expansion has been established for a model with two symmetric tensors \cite{Gurau:2017qya}. In this paper we will build on the method introduced there to establish the $1/N$ expansion for 
symmetric traceless and antisymmetric tensors. We prove the conjecture of \cite{Klebanov:2017nlk}, namely that in both the symmetric traceless and the antisymmetric cases, tensor models with a tetrahedral interaction support a $1/N$ expansion dominated by melonic graphs.
 
A crucial difference between the cases we treat in this paper and the two-tensor model treated in \cite{Gurau:2017qya}
is that we now need to deal with tadpole graphs. It turns out that the tadpoles are quite problematic: individual graphs containing tadpoles violate the 
maximal scaling in $N$. For both the antisymmetric and the symmetric traceless models the tadpoles can be eliminated by a \emph{partial resummation} of the perturbative series.
The resummation of all the tadpoles brings a nontrivial cancellation at leading order in $N$. 
Crucially such cancellation does not occurr for the  
symmetric model with no tracelessness condition: in order to offset the contribution of the tadpoles in that case one must use a different rescaling of the coupling constant with which
the  melonic graphs are strictly suppressed. The difference between the symmetric cases with and without tracelessness condition can be understood in terms of the irreducible and reducible nature of their respective representation of the $O(N)$ group, and the applicability of Schur's Lemma (see Lemma \ref{lemma:2-point-N}).

In quantum field theory the tadpoles are subtracted by Wick ordering which in turn requires the addition of counter terms to the action. We will show that the original 
theories are exactly equivalent to theories with a Wick ordered interaction but with a renormalized covariance. One then checks explicitly, and this is a nontrivial check, 
that the renormalized covariance has a well defined large $N$ limit and is itself a series in $N^{-1/2}$ (it is this step that fails for the symmetric, non traceless case).
In this new formulation the tadpoles are subtracted, and one can now show in the resulting theory, which does not have any more tadpoles, the correct scaling bound in $N$.

{\bf Outline of the paper.}
In Section~\ref{sec:models} we introduce the models and the main results of our paper: Theorem~\ref{thm:main}, which establishes the existence of the $1/N$ expansion, and Theorem~\ref{thm:LO}, which characterizes its leading order. The rest of the paper is dedicated to their proof. In Section~\ref{sec:expansion} we present various aspects of the perturbative expansion and we provide a glimpse of the proof; in particular, in Section~\ref{sec:graphs-def} we define the two types of graphs (or more precisely maps) that are used to describe the perturbative expansion, and we derive expressions for their associated amplitudes, while  in Section~\ref{sec:graphs-ex} we discuss several important classes of graphs that play an important role in the proof. We close the section with Proposition~\ref{prop:moves}, which is an essential step in the proof of Theorem~\ref{thm:main}. The proof of Proposition~\ref{prop:moves} is the most tedious part of the paper, and we postpone it to Section~\ref{sec:proof1}. In Section~\ref{sec:subtr} we show how tadpoles and melons can be resummed by means of the Wick ordering trick; their resummation leads us to consider graphs with no tadpoles and no melons, for which we prove Proposition~\ref{prop:main}, the last ingredient for the proof of Theorem~\ref{thm:main}. In Section~\ref{sec:LO}, we finally prove Theorem~\ref{thm:LO}, which states that the $1/N$ expansion is dominated by melonic diagrams. In the four Appendices we collect some useful ingredients and side results.

\newpage

\section{The models and the main results}
\label{sec:models}

We will treat the symmetric traceless and  antisymmetric cases simultaneously. 
Let us first consider real tensors of rank $3$ having no symmetry property under permutation of their indices, thus having $N^3$ independent components, and
transforming in the direct product of three copies of the fundamental representation of the orthogonal group $O(N)$:
\be \label{eq:gen-tensor}
T_{a_1 a_2 a_3} \rightarrow T'_{a_1a_2 a_3} = (O \cdot T)_{a_1a_2 a_3} := \sum_{b_1,b_2,b_3} O_{a_1b_1} O_{a_2b_2} O_{a_3b_3} T_{b_1b_2b_3} \;,  \qquad O \in O(N)\;.
\ee
Although similar to the models built out of the fundamental representation of $O(N)^{\otimes 3}$ \cite{RTM,expansion1,critical,uncoloring,Carrozza:2015adg,Klebanov:2016xxf}, a tensor model built upon \eqref{eq:gen-tensor} differs from them in an important aspect:
all the indices of the tensor transform with the same orthogonal matrix, hence one can build invariants in which an index in the first position on a tensor is contracted with an index in another position, say the third position, on another (or the same) tensor.

We denote ${\bf 1}$ the identity operator in the space of tensors ${\bf 1}_{a_1a_2a_3,b_1b_2b_3}= \delta_{a_1b_1} \delta_{a_2b_2} \delta_{a_3b_3}$, and:
\[
T {\bf 1} T  \equiv \sum_{\genfrac{}{}{0pt}{}{a_1,a_2,a_3}{b_1,b_2,b_3}}  T_{a_1a_2a_3}  {\bf 1}_{a_1a_2a_3,b_1b_2b_3}  T_{b_1b_2b_3}  \;,\qquad 
  \partial_T {\bf 1} \partial_T \equiv \sum_{\genfrac{}{}{0pt}{}{a_1,a_2,a_3}{b_1,b_2,b_3}} \frac{\partial}{\partial T_{a_1a_2a_3}} {\bf 1}_{a_1a_2a_3,b_1b_2b_3} \frac{\partial}{\partial T_{b_1b_2b_3}}  \;.
\] 
 
The Gaussian integral with covariance ${\bf 1}$, can be represented as a differential operator \cite{RTM,Brydges:2014nba,salmhofer1999renormalization}; for example, the free 2-point function is written as:
\begin{align*}
 & \Braket{ T_{a_1a_2a_3} T_{b_1b_2b_3}  }_0    =  \int [dT] \; e^{-\frac{1}{2} T {\bf 1} T}  \; T_{a_1a_2a_3} T_{b_1b_2b_3} = 
  \left[ e^{\frac{1}{2} \;  \partial_T {\bf 1} \partial_T    }      \;  T_{a_1a_2a_3} T_{b_1b_2b_3}
\right]_{T=0} =   {\bf 1}_{a_1a_2a_3,b_1b_2b_3} \;,
\end{align*}
and its full contraction gives $\Braket{T {\bf 1} T}_0  =  N^3$. The generic tensor model with tetrahedral interaction is defined by the action:
\begin{align} \label{eq:action}
 S(T) = & \frac{1}{2} \sum_{a_1 , a_2 , a_3}  T_{a_1a_2a_3}   T_{a_1a_2a_3} - \frac{  \lambda  }{ 4 N^{3/2}}   
 \sum_{a_1\dots a_6}  T_{a_1a_2a_3}   T_{a_3a_4a_5}    T_{a_5 a_2 a_6}  T_{a_6a_4 a_1} \;, 
\end{align}
where the sign of the coupling constant follows the usual conventions in matrix models \cite{DiFrancesco:1993nw}.
The normalization of the interaction is well known in tensor models and it is the only one which can lead to an interesting large $N$ 
limit \cite{RTM,Carrozza:2015adg,Klebanov:2017nlk}. The partition function, the free energy and its (appropriately normalized) first derivative are:
\begin{align}\label{eq:modelgen}
 Z_{\bf 1}(\lambda) = &    \int [ dT] \; e^{-S(T)} = \left[ e^{\frac{1}{2} \;   \partial_T {\bf 1} \partial_T   }      \; e^{  \frac{  \lambda  }{4 N^{3/2}}   
 \sum_{a_1\dots a_6}  T_{a_1a_2a_3}   T_{a_3a_4a_5}    T_{a_5 a_2 a_6}  T_{a_6a_4 a_1} } \right]_{T=0} \;, \crcr
 \ln Z_{\bf 1}(\lambda) & = \ln\bigg\{ \int [ dT] \; e^{-S(T)} \bigg\} \; , \qquad F_{\bf 1}(\lambda) = \frac{4}{N^3} \lambda\partial_{\lambda} \ln Z_{\bf 1}(\lambda)  \;.
\end{align}
In order to simplify the combinatorics we will deal below with $F_{\bf 1}(\lambda)$.

\

Unlike the fundamental representation of $O(N)^{\otimes 3}$, which is irreducible, the direct product of three copies of the fundamental representation of $O(N)$ is a reducible representation.
Therefore, a generic tensor transforming as in \eqref{eq:gen-tensor} can be decomposed in irreducible components. This is achieved by removing its traces, and decomposing the rest in terms of irreducible representations
of the symmetric group $ \mathfrak{S}_3$, which commutes with the action of $O(N)$. The result is a decomposition of $T_{a_1a_2 a_3}$ into the following irreducible objects: a completely symmetric and traceless tensor, 
a completely antisymmetric one, two tensors with mixed symmetry, and three lower-rank tensors (the traces, which in our rank-3 case simply correspond to vector objects). 
Under such a decomposition, the quadratic part of the action \eqref{eq:action} partially diagonalizes\footnote{Each of the two degenerate sectors of the two tensors with mixed symmetry and of the three traces remain internally mixed.}, while the quartic interaction leads to a mixing between the various irreducible components.
From the point of view of model building, the traces lead to hybrid models, mixing different ranks, such as the models studied in \cite{Halmagyi:2017leq}. The tensors with mixed-symmetry might be interesting 
from the point of view of the multi-matrices interpretation of rank-3 tensors \cite{Ferrari:2017ryl,Ferrari:2017jgw,Azeyanagi:2017drg,Azeyanagi:2017mre}. However, sticking to a purely rank-3 tensor point of view, it makes sense to restrict oneself 
to the completely symmetric or antisymmetric components. The two can mix in the quartic interaction resulting in a term with two antisymmetric and two symmetric tensors; however, for the sake of simplicity, in this
paper we will not consider possible mixing between irreducible components, and 
we will deal with either  \emph{completely symmetric and traceless} or  \emph{completely antisymmetric} 
tensors.

We denote $\bA$ the orthogonal projector on antisymmetric tensors:
\begin{align}\label{eq:A}
& \bA_{a_1a_2a_3, b_1b_2b_3} = \frac{1}{3!} \sum_{\sigma\in \mathfrak{S}_3} \epsilon(\sigma) \prod_{i=1}^3 \delta_{a_i b_{ \sigma(i)} }  \\
& \qquad =   \frac{1}{3!} \delta_{a_1b_1} ( \delta_{a_2b_2} \delta_{a_3 b_3} - \delta_{a_2b_3} \delta_{a_3b_2}  ) + 
        \frac{1}{3!}  \delta_{a_1b_2} ( - \delta_{a_2b_1} \delta_{a_3 b_3} +\delta_{a_2b_3} \delta_{a_3b_1}  )  +
      \frac{1}{3!} \delta_{a_1b_3} (\delta_{a_2 b_1} \delta_{a_3 b_2} -  \delta_{a_2 b_2} \delta_{a_3 b_1} )  \;, \nonumber
\end{align}
and $\bS$ the orthogonal projector on symmetric traceless tensors \cite{Klebanov:2017nlk}:
\begin{align}\label{eq:S}
& \bS_{a_1a_2a_3, b_1b_2b_3} = \frac{1}{3!} \bigg[ \sum_{\sigma\in \mathfrak{S}_3}  \prod_{i=1}^3 \delta_{a_i b_{ \sigma(i)} }  -\frac{2}{N+2}  
   \sum_{i,j \in \{1,2,3\} } \delta_{a_ib_j}  \prod_{k,l \neq i } \delta_{a_k a_l} \prod_{r,s\neq j} \delta_{b_rb_s}
\bigg]\crcr
& =  \frac{1}{3!} \bigg[ 
  \delta_{a_1b_1} ( \delta_{a_2b_2} \delta_{a_3 b_3} + \delta_{a_2b_3} \delta_{a_3b_2}  ) + 
         \delta_{a_1b_2} (  \delta_{a_2b_1} \delta_{a_3 b_3} +\delta_{a_2b_3} \delta_{a_3b_1}  )  +
      \delta_{a_1b_3} (\delta_{a_2 b_1} \delta_{a_3 b_2} +  \delta_{a_2 b_2} \delta_{a_3 b_1} )  \crcr
 & \qquad -\frac{2}{N + 2} \bigg( 
   \delta_{a_1b_1} \delta_{a_2a_3} \delta_{b_2b_3} + \delta_{a_1b_2} \delta_{a_2a_3} \delta_{b_1b_3} + \delta_{a_1b_3} \delta_{a_2a_3} \delta_{b_1b_2} + (a_1 \leftrightarrow  a_2) + (a_1 \leftrightarrow  a_3)
 \bigg)
 \bigg] \;. 
\end{align}

We will denote generically $\bP = \bA, \bS $ one of the two projectors.
The tensor models for symmetric traceless and antisymmetric tensors with tetrahedral interaction are obtained from the generic model of  Eq.\eqref{eq:modelgen} by allowing the propagation
of only the antisymmetric (respectively symmetric traceless) modes of the tensor\footnote{This is equivalent to giving an infinite mass to the orthogonal modes $({\bf 1} - \bP)T$ of the tensor.}:
\begin{align}\label{eq:crucial}
  F_{\bP}(\lambda) = \frac{4}{N^3} \lambda\partial_{\lambda} \ln \bigg\{
     \left[ e^{\frac{1}{2} \;   \partial_T \bP \partial_T   }      \; e^{  \frac{  \lambda  }{4 N^{3/2}}   
 \sum_{a_1\dots a_6}  T_{a_1a_2a_3}   T_{a_3a_4a_5}    T_{a_5 a_2 a_6}  T_{a_6a_4 a_1} } \right]_{T=0}
  \bigg\} \;.
\end{align}

Because only the projected modes $\bP T$ propagate, one can either take Eq.\eqref{eq:crucial} as definition and consider that
the tensor $T$ still has no symmetry property under permutation of its indices, or one can change variables to $P = \bP T$ and write equivalently:
\begin{align}\label{eq:crucial1} 
&  F_{\bP}(\lambda)  = \frac{4}{N^3} \lambda\partial_{\lambda} \ln \bigg\{
     \left[ e^{\frac{1}{2} \;   \partial_P \bP \partial_P  }      \; e^{  \frac{  \lambda  }{4 N^{3/2}}   
 \sum_{a_1\dots a_6}  P_{a_1a_2a_3}   P_{a_3a_4a_5}    P_{a_5 a_2 a_6}  P_{a_6a_4 a_1} } \right]_{P=0}
  \bigg\} \crcr
& \frac{\partial}{\partial P_{a_1a_2a_3}} P_{b_1b_2b_3}  \equiv \bP_{a_1a_2a_3,b_1b_2b_3} \;,
\end{align}
where this time the tensor $P$ is antisymmetric or symmetric traceless. Observe that the second line is a \emph{definition}.
The $1/N$ expansion of the symmetric traceless and respectively antisymmetric tensor model in rank $3$ with tetrahedral interaction is encoded in the following theorem, the main result of this paper.
\begin{theorem}\label{thm:main}
We have (in the sense of perturbation series):
\[
 F_{\bP}(\lambda) = \sum_{\omega \in \mathbb{N}/2} N^{-\omega} F_{\bP}^{(\omega)} (\lambda) \;.
\]
\end{theorem}
\begin{proof}
 This follows from Eq.~\ref{eq:rewrite} and Proposition~\ref{prop:main}.
 
\end{proof}

In a second stage, we will prove that the model is dominated by melon diagrams (which we will introduce in the next section). 
\begin{theorem}\label{thm:LO}
The leading order contribution $F_{\bP}^{(0)} (\lambda)$ is a sum over melonic stranded maps. 
\end{theorem}
\begin{proof}
This follows from Eq.~\ref{eq:rewrite}, Proposition \ref{propo:LO-ring}, and the fact that the partially resumed covariance $K(\lambda,N)$ introduced in Section \ref{sec:subtr} reduces to a sum over $2$-point melonic graphs in the large $N$ limit.
\end{proof}

The rest of this paper is dedicated to proving these two theorems. Similar expansions can be obtained by the usual means for the expectation of invariant observables of the model.

\newpage

\section{The perturbative expansion}
\label{sec:expansion}
    
We review the perturbative expansion of the models.
    
\subsection{Feynman and stranded graphs}
\label{sec:graphs-def}

The perturbative expansion of $F_{\bP}(\lambda)$ is obtained by:
\begin{itemize}
 \item[--] Taylor expanding in $\lambda$ and computing the Gaussian integrals (see Appendix~\ref{app:gauss}  for a brief discussion), which yields a sum over four-valent \emph{combinatorial maps} or embedded graphs;
 \item[--] taking the logarithm, which gives a sum over \emph{connected} combinatorial maps;
 \item[--] applying the operator $4\lambda \partial_{\lambda}$ which leads to \emph{rooted} connected combinatorial maps. A rooted map is a map with a halfedge on a vertex marked with an incoming arrow.
\end{itemize}
Ignoring for a second the scaling with $N$, at first order in $\lambda$ we have $F_{\bP}(\lambda) = \lambda \left[ e^{\frac{1}{2}    \partial_T \bP \partial_T     } TTTT  \right]_{T=0}$,
and the three corresponding rooted, connected, combinatorial maps are represented in  Fig.~\ref{fig:map1}.

\begin{figure}[htb]
 \begin{center}
 \includegraphics[scale=.6]{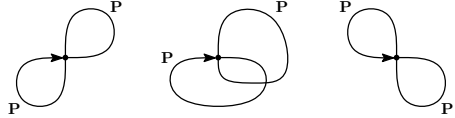}  
 \caption{First order contributions to $F_{\bP}(\lambda)$} \label{fig:map1}
 \end{center}
 \end{figure}

We denote the four valent, rooted, connected, combinatorial maps by ${\cal M}$. Each unlabeled map comes with combinatorial weight $1$
(this is why we chose to study $F_{\bP}(\lambda)$ \cite{DiFrancesco:1993nw,RTM,critical,review}):
\be\label{eq:pertexp1}
F_{\bP}(\lambda) = \sum_{ {\cal M} \text{ connected, rooted} } \lambda^{V(\cal M)} A( {\cal M} ) \;,
\ee
where $V(\cal M) $ denotes the number of vertices of ${\cal M}$ and $A( {\cal M} )$ is the amplitude of the map ${\cal M}$ which we now define.

We label the vertices of the map by $v\in \cM$, and the half edges $i_v$, with $i=0,1,2,3$. Every half edge has an associated ordered triple of indices 
$\ba^{i_v} =( a^{i_v}_{0_v}, a^{i_v}_{1_v} ,  a^{i_v}_{2_v} , a^{i_v}_{3_v} ) \setminus (a^{i_v}_{i_v})$. 
The edges of the map $e\in \cM$ will be denoted by the couple of labels of the half edges they connect: $e=(i_v,j_w)$ is the edge connecting the half edge $i_v$ of the vertex $v$
with the half edge $j_w$ on the vertex $w$. The amplitude $A( {\cal M} )$ of $\cM$ is:
\[
A(\cM) = N^{ -3 -\frac{3}{2} V(\cM) }\sum_{  a } \left( \prod_{v\in \cM} \;  \prod_{  i<j } \delta_{a^{i_v}_{j_v} a^{j_v}_{i_v } } \right) \prod_{e = (i_v,j_w) \in \cM} \bP_{ \ba^{i_v} , \ba^{j_w} } \;,
\]
where $\sum_a$ denotes the sum over all the indices $a$. 

We call, for obvious reasons, the expansion in Eq. \eqref{eq:pertexp1} the expansion in 
\emph{Feynman maps}. This expansion has the drawback that each $\bP_{ \ba^{i_v} , \ba^{j_w} }$ is a sum of terms (six in the antisymmetric case and fifteen in the symmetric traceless one), and 
these terms have different scaling with $N$.

\

It is convenient to pass from the Feynman expansion to a more detailed expansion in \emph{stranded maps} $\cS$. A stranded map is a combinatorial map with a choice of one (of six respectively fifteen) terms 
in Eq.\eqref{eq:A},\eqref{eq:S} for every edge. The six terms common to the two cases lead to \emph{unbroken} edges and the nine extra terms in \eqref{eq:S} to \emph{broken} edges. 
We call each pair of  indices contracted in the interaction vertex a \emph{corner}. The vertex is then represented as a four-valent stranded vertex with six corners, as in Fig.~\ref{fig:vertex}. 
Observe that four corners appear as genuine corners in the graphical representation of Fig.~\ref{fig:vertex}. However, the reader should bear in mind that strands which go through 
the vertex of Fig.~\ref{fig:vertex} also have a corner each (this could be represented by a dot on the middle of the strands, but we 
avoid doing this so as not to overcharge the figures). 
The stranded edges (see again Fig.~\ref{fig:vertex}) connect the corners via strands. For an unbroken edge all the strands 
traverse and connect corners at the two ends. For the broken edges, a pair of corners is connected by a strand at each end of the edge,
and one strand traverses (see Fig.~\ref{fig:vertex}).
 \begin{figure}[htb]
 \begin{center}
 \includegraphics[scale=.8]{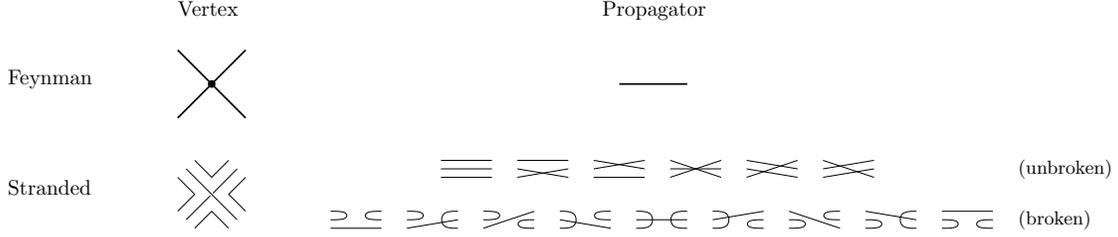}  
 \caption{Vertex and propagator of the model.} \label{fig:vertex}
 \end{center}
 \end{figure}
 
 We introduce some more notation. The unbroken edge $e = (i_v,j_w)_{\sigma^e}$ identifies pairwise the indices in the triple $\ba^{i_v}$ with a permutation $\sigma^e$ of the indices in the triple
 $\ba^{j_w}$,  $  \ba^{i_v} = \sigma^e(\ba^{j_w}  )$
 and brings a sign $\epsilon(\sigma^e)$ (which is always trivial in the symmetric case). A broken edge 
 $e = (i_v,j_w)_{k^e_v l^e_w}$
 identifies the index $a^{i_v}_{k^e_v }$ with the index $a^{j_w}_{l^e_w }$ and then identifies the 
 remaining pair of indices $k'_vk''_v$ on $i_v$ and $l'_wl''_w$ on $j_w$. The perturbative expansion now becomes:
\begin{align*}
 F_{\bP}(\lambda) & = \sum_{ \cS  \text{ connected, rooted} } \lambda^{ V (\cS)} A( \cS ) \;,  \crcr
 A(\cS) & = N^{ -3 -\frac{3}{2} V(\cS) }\sum_{  a } \left( \prod_{v\in \cS} \;  \prod_{  i<j } \delta_{a^{i_v}_{j_v} a^{j_v}_{i_v } } \right) 
\left(   \prod_{e = (i_v,j_w)_{\sigma^e} \in \cS}^{e \text{ unbroken}}  \frac{\epsilon(\sigma^e ) }{3!}  \; \delta_{   \ba^{i_v}  \sigma^e(\ba^{j_w}  )    }  \right) \crcr
& \qquad \qquad \left(  \prod_{e = (i_v,j_w)_{k^e_vl^e_w} \in \cS}^{e \text{ broken}}  \frac{ (- 2) }{ 3! (N+2) } \; \delta_{ a^{i_v}_{k^e_v} a^{j_w}_{l^e_w } } 
\delta_{ a^{i_v}_{k'_v} a^{i_v}_{ k''_v }} \delta_{ a^{j_w}_{l'_w} a^{j_w}_{l''_w}}   \right) \;.
\end{align*}
There are six choices of permutations $\sigma^e$ for the unbroken edges and nine choices of pairs $k^e_vl^e_w$ for the broken edges.
In a stranded map, the strands close into \emph{faces}. When computing the amplitude of $\cS$ 
one obtains a free sum per face.
We denote $F(\cS)$ the number of faces of $\cS$, $B(\cS)$ and $U(\cS)$ the number of broken and unbroken edges of $\cS$
and $ \epsilon(\cS) = (-1)^{ B(\cS) } \prod_{e\in \cS}^{e \text{ unbroken}} \epsilon(\sigma^e)  $ the sign of $\cS$. The amplitude of a stranded map is then:
\[
 A(\cS)  =  \left(  \frac{  \epsilon(\cS)  }{3!^{U(\cS) + B(\cS) } } \; \frac{ 1 }{ \left(  1 + \frac{2}{N} \right)^{ B(\cS)} }  \right) \; \; N^{  - \omega(\cS)} \;, 
\]
where we define the \emph{degree} of a connected map:
\be \label{eq:degree}
\boxed{ \omega(\cS) = 3 +  \frac{3}{2} V(\cS)  + B(\cS) - F(\cS)  } \;,
\ee
and the degree of a disconnected map as the sum of the degrees of its connected components. 
A priori the degree can be any half integer (positive, zero or negative), and the perturbative expansion of 
$F_{\bP}(\lambda)$ writes:
 \be\label{eq:perte1}
 F_{\bP}(\lambda)  = \sum_{ \cS  \text{ connected, rooted} } \lambda^{ V (\cS)} \left(  \frac{  \epsilon(\cS)  }{3!^{U(\cS) + B(\cS) } } \; \frac{ 1 }{ \left(  1 + \frac{2}{N} \right)^{ B(\cS)} }  \right)
 \; \; N^{  - \omega(\cS)} \;.
\ee

\paragraph{\it The main difficulty.}   Eq.~\eqref{eq:perte1} is an expansion very similar to the one we aim for in Theorem~\ref{thm:main}. However, 
the subtlety comes from the following fact: it is \emph{not true} that the 
 degree of any stranded map is non negative. In order to prove Theorem~\ref{thm:main}
 we first need to improve the naive perturbative expansion by performing a partial resummation of the perturbative expansion. This 
is done in Section~\ref{sec:subtr} and is the first step in establishing the $1/N$ expansion.

 \subsection{Examples of graphs}
 \label{sec:graphs-ex}
 
By abuse of language, but in keeping with the physics literature, we will sometimes refer to the maps as \emph{graphs} (stranded graphs, etc.). 
Several classes of maps will play a role in the sequel.

\subsubsection{Melons and tadpoles}\label{sec:cores}

At low orders we have two interesting graphs.
\begin{definition}
 We call a \emph{tadpole} an amputated\footnote{Amputated graphs have two external half edges hooked to two vertices.} one-particle-irreducible two-point Feynman graph with one vertex. We call \emph{melon}
 an amputated one-particle-irreducible two-point Feynman graph with two vertices.
\end{definition}

As a function of the embedding, there are a priori $3$ tadpoles and $6$ melons (see Fig.~\ref{fig:meltad}, the top two rows). Each tadpole leads to $6^3$ or $15^3$ stranded tadpoles, 
and each melon to $6^5$ or $15^5$ stranded melons (where we also take into account the choices of stranded external half edges).
   \begin{figure}[htb]
 \begin{center}
 \includegraphics[scale=.6]{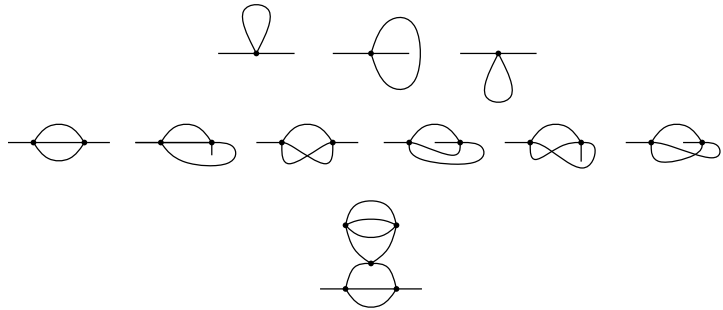}  
 \caption{The tadpole, the melon and a melon-tadpole graph.}\label{fig:meltad}
 \end{center}
 \end{figure}

\paragraph{Bad tadpoles.} Some stranded tadpoles can be used to build examples of stranded graphs with negative degree. We will call these tadpoles \emph{bad tadpoles}, and they  
have the following features (see Fig.~\ref{fig:badtad}):
 \begin{itemize}
  \item they have one internal face of length $1$,
  \item they have two external strands which return on the same external half edge, and one external strand which is transmitted from one half edge to the other.
 \end{itemize}
   \begin{figure}[htb]
 \begin{center}
 \includegraphics[scale=.5]{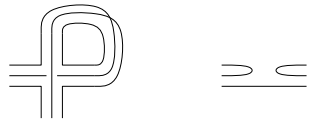}  
 \caption{A bad tadpole and its external strands.}\label{fig:badtad}
 \end{center}
 \end{figure}
 
 The problem with the bad tadpoles is that one can build a chain of $V$ bad tadpoles, as represented in Fig.~\ref{fig:chainbadtad}. The chain will have $V$ vertices, and $2V$ faces (each bad tadpole
 brings one  internal face of length $1$ and closes one external strand into a face). Such a chain will always bring a contribution:
$3V/2 - 2V = -V/2$ to the degree of a graph, hence inserting long enough chains of bad tadpoles on the edges, the degree of any graph ultimately becomes negative. 
   \begin{figure}[htb]
 \begin{center}
 \includegraphics[scale=.7]{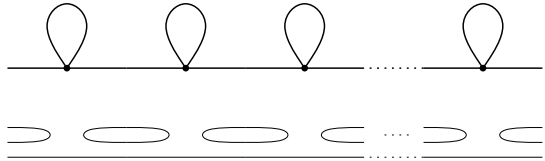}  
 \caption{A chain of bad tadpoles.}\label{fig:chainbadtad}
 \end{center}
 \end{figure}

\paragraph{Elimination of the tadpoles.}
It turns out that for both the antisymmetric and the symmetric traceless cases all the tadpoles (including the bad ones) are eliminated by a partial resummation of the perturbative series.
In both cases the resummation of all the tadpoles brings an unexpected cancellation at leading order in $N$. This is quite nontrivial and in particular does not work for the  
symmetric model with no tracelessness condition\footnote{It is for this reason that simply symmetric tensors do not support the same kind of $1/N$ expansion. In order to offset the contribution of bad tadpoles
in that case one must use a different rescaling of the coupling constant: in our notation one needs 
to keep $\lambda N^{1/2}$ fixed  when sending $N$ to infinity. This leads to a different $1/N$ expansion in which
the  melonic graphs are strictly suppressed.}.

In quantum field theory, it is well known that the tadpoles are subtracted by Wick ordering. This requires the addition of counter terms. 
We will show below by \emph{adding and subtracting} a counterterm that the original theory, defined by 
\eqref{eq:crucial} is exactly equivalent to a theory with a Wick ordered interaction but with a renormalized covariance. In this new theory 
the tadpoles are subtracted, and one can show that the graphs with no tadpoles have non negative degree. One then checks explicitly
(and this is a nontrivial check) that the renormalized covariance has a well defined large $N$ limit and is itself a series in $N^{-1/2}$. This proves Theorem~\ref{thm:main}.

The same rewriting can of course be performed in the symmetric (non traceless) case but, unless one rescales further the coupling constant, the trace modes
develop an instability in the large $N$ limit, and the covariance can not be renormalized (the would be renormalized covariance is 
the sum of a divergent series).

\paragraph{Melon-tadpole graphs.} A further subtlety comes from the fact that one must subtract not only the tadpoles, but also the melons. The degree defined in Eq.~\eqref{eq:degree}
is unchanged by the insertion of a melon on any edge. Thus a graph can be constructed by inserting a melon on the internal edge of a 
bad tadpole, see Fig.~\ref{fig:badtadmel}. Chaining this new graph one can build an arbitrary chain with no tadpoles, but which has the same effect as the chain of bad tadpoles.

   \begin{figure}[htb]
 \begin{center}
 \includegraphics[scale=.7]{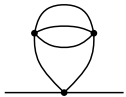}  
 \caption{An example of dangerous graph obtained by inserting a melon in a tadpole.}\label{fig:badtadmel}
 \end{center}
 \end{figure}

It follows that one needs to perform the explicit resummation and subtraction of both the tadpole and the melon graphs and 
find an equivalent formulation of the theory in terms of graphs having neither tadpoles nor melons.

\begin{definition}
The lowest orders melon-tadpole graphs are the melon and the tadpole.
A melon-tadpole graph is an amputated two-point graph which can be obtained from 
a lower order melon-tadpole graph by inserting a melon or a tadpole on an edge (see Fig.\ref{fig:meltad}, bottom row). 
\end{definition}

If a graph has neither melons nor tadpoles then it does not have any melon-tadpole subgraphs.
An edge in a two-point graph is a one-particle reducibility (1PR) edge if the graph disconnects by cutting it. A graph with no 1PR edges is called one-particle irreducible (1PI).
Melon-tadpole graphs can be defined alternatively as the two-point graphs such that their 1PI components (obtained by cutting all the 1PR edges)
are either a melon or a tadpole with arbitrary melon-tadpole insertions on the internal edges. 
It can be shown following \cite{GurSch} that if two melon-tadpole graphs are not totally disjoint then their union is a melon-tadpole graph.
It follows that melon-tadpole subgraphs in a graph $\cS$ can be extended maximally to non overlapping melon-tadpole graphs and can be eliminated simultaneously 
to pass to a melon-tadpole free \emph{core}.

\subsubsection{Dipoles, triangles and chains} 

A graph $\hat S$ which has neither melons nor tadpoles is sometimes called \emph{melon-tadpole free}.

\begin{definition}\label{def:dipole}
A $2$--\emph{dipole} (dipole for short) consists in 2 vertices connected by exactly 2 edges (which we call internal edges).  
A dipole has 4 external half edges, two on each vertex.

A \emph{triangle} is a subgraphs made out of $3$ vertices, $3$ edges connecting them pairwise, and $6$ external half edges (two per vertex). 
\end{definition}

If a pair of external half edges of a dipole is joined into an edge one obtains either a tadpole or a melon. If a pair of external half edges of a
triangle is connected into an edge one obtains either a tadpole or a dipole. Thus 
 If $\hat \cS$ is melon-tadpole free and has a dipole, then no pair of the four external half edges of the dipole can be joined into an edge of $\hat \cS$
 and if $\hat \cS$ is melon-tadpole free and has no dipole but has a triangle, no pair of the six external half edges of the triangle can be joined into an edge of $\hat \cS$.

\begin{definition}
   A \emph{chain} is a melon-tadpole free four-point graph which becomes a melon-tadpole graph by connecting one pair of external halfedges into an edge.
   A chain is \emph{proper} if it has at least three vertices\footnote{Observe that according to this definition, a vertex or a $2$-dipole are also chains. They are however not proper.}. A chain is maximal if it is not a subgraph of a larger chain. 
\end{definition}
\begin{figure}[htb]
 \begin{center}
 \includegraphics[scale=.8]{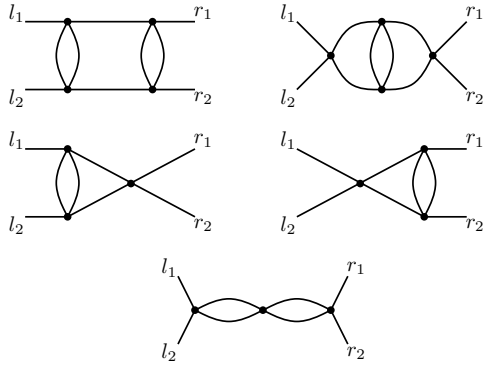}  
 \caption{The shortest proper chains.}\label{fig:minpropchain}
 \end{center}
 \end{figure}

The proper chains with three and four vertices are presented in Fig.~\ref{fig:minpropchain}. 
The proper chains can be extended maximally: starting from one of the chains in Fig.~\ref{fig:minpropchain}, 
one may try to extend the chain (say towards the right). To this end, one checks whether the 
right external half edges $r_1$ and $r_2$ of the chain are joined to the left external half edges $l_1'$ and $l_2'$ of a dipole. If they are, then the chain can be extended, and there are two possible
cases (see Fig.~\ref{fig:dipoles}):
\begin{itemize}
 \item  $l_1'$ and $l_2'$ are incident one to each vertex of the dipole, hence the dipole is vertical;
 \item  $l_1'$ and $l_2'$ are incident to the same vertex of the dipole, hence the dipole is horizontal.
\end{itemize}
One may then repeat this procedure until a maximal number of dipoles have been added to the chain, and a \emph{maximal proper chain} has been identified. 

\begin{figure}[htb]
 \begin{center}
 \includegraphics[scale=.8]{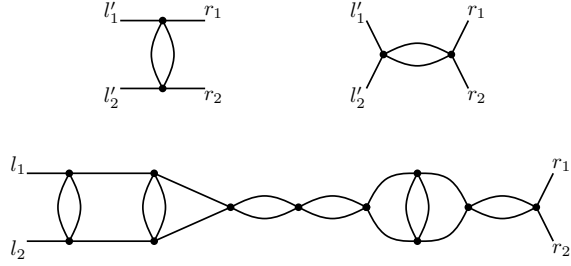}  
 \caption{Horizontal dipole, vertical dipole and a maximal chain.}\label{fig:dipoles}
 \end{center}
 \end{figure}
 
\subsubsection{End graphs} 

We denote $F_q(\cS)$ the number of faces with $q$ corners (i.e. of length $q$) of a graph $\cS$. We call the faces of length $1$, $2$ or $3$ \emph{short}. 

\begin{proposition}\label{prop:end}
We call \emph{end graphs} the stranded graphs $\cS$ such that:
 \begin{itemize}
  \item either $\cS$ has no vertices, that is it is a ring graph consisting in only one edge closed onto itself, 
  \item or $\cS$ has no short faces, $F_1(\cS) = F_2(\cS) = F_3(\cS) = 0$.
 \end{itemize}
 End graphs have non negative degree.
\end{proposition}
\begin{proof}
As a vertex contributes $6$ corners to the faces we have:
\[
F_1(\cS) + 2F_2(\cS) + 3 F_3(\cS) + 4F_4(\cS) + \dots = 6V(\cS) \;,
\]
therefore, if a connected graph $\cS$ has no short faces, then $3 V(\cS) \geq 2 F(\cS)$ and the degree $\omega(\cS)$ is at least $3$. 
The ring graphs have no vertex and at most $3$ faces, hence non negative degree.

\end{proof}
 
\subsubsection{Graphs with no tadpoles, no melons and no broken edges} 
 
We will encounter graphs $G$ with no tadpoles, no melons and no broken edges. 
Observe that for such graphs the faces of length two must belong to dipoles, and the ones of length three to triangles. 
 The main technical result of this paper is the following Proposition.
 
\begin{proposition}\label{prop:moves}
 If a graph $G$ with no tadpoles, no melons and no broken edges has a dipole, then there exists a graph $G'$ such that:
     \begin{itemize}
         \item $\omega(G) \ge \omega(G')$,
         \item $V(G)> V(G')$,
         \item $G'$ has no tadpoles, no melons and no broken edges.
\end{itemize}
The same holds if $G$ has no dipole, but it has a triangle.
\end{proposition}
\begin{proof}
This follows from Lemmata~\ref{lem:dipole} and~\ref{lem:tripole} in Section~\ref{sec:proof1}.

\end{proof}

 \newpage

\section{Subtraction}
\label{sec:subtr}

We now show that the models for both symmetric traceless and antisymmetric tensors in Eq.~\eqref{eq:crucial1} are equivalent 
to models with melons and tadpoles subtracted and renormalized covariance.
Distinguishing the fields in the interaction in Eq.~\ref{eq:crucial1} (and considering that repeated indices are summed) we have in 
both cases\footnote{The second equality involves only permutations of indices and the last one is just a relabeling.}: 
\begin{align*}
 & P^0_{a_1a_2a_3}   P^1_{a_3a_4a_5}    P^2_{a_5 a_2 a_6}  P^3_{a_6a_4 a_1}  = 
 P^0_{a_1a_2a_3}  P^2_{a_5 a_2 a_6}   P^1_{a_3a_4a_5}  P^3_{a_6a_4 a_1}  \crcr 
 & \qquad  =P^0_{a_1a_3a_2}  P^2_{a_2 a_6 a_5}        P^1_{a_5a_3a_4}  P^3_{a_4a_6 a_1} = 
   P^0_{a_1a_2a_3}  P^2_{a_3 a_4 a_5}   P^1_{a_5a_2a_b}  P^3_{a_6a_4 a_1} \;,  
\end{align*}
that is, the interaction is invariant under arbitrary permutations of the fields hence the amplitude of the Feynman map $\cM$ is insensitive to the embedding of $\cM$. 

\

Let us consider a theory with covariance $K\bP$, where $K$ is a real number. The self energy (amputated  one-particle-irreducible two-point function) up to order $\lambda^2$ is:
\begin{align*}
 &  \Sigma^{(2)}_{a_1a_2a_3,b_1b_2b_3}  = 3 \frac{\lambda K}{N^{3/2}} \sum_{c} \bP_{a_1a_2a_3, c_1c_2c_3} \bP_{c_3c_4c_5, c_5c_2c_6} \bP_{c_6 c_4 c_1,b_1b_2b_3}   \crcr
 & \qquad   + 6 \frac{\lambda^2 K^3}{N^3}
     \sum_{c,d} \bP_{a_1a_2a_3, c_1 c_2c_3}    \bP_{c_3 c_4 c_5, d_3 d_4 d_5 } \bP_{c_5c_2 c_6 , d_5 d_2 d_6 } \bP_{c_6c_4c_1 ,  d_6 d_4 d_1  } \bP_{ d_1 d_2 d_3 , b_1b_2b_3} \;.
\end{align*}
Using the explicit formulae of the two projectors in Eq.\eqref{eq:A}, \eqref{eq:S}, a short computation shows that:
\[
 \Sigma^{(2)}_{a_1a_2a_3,b_1b_2b_3} =  \bigg(  \lambda K f^{\bP}_1  +    \lambda^2 K^3 f_2^{\bP}   \bigg) \bP_{a_1a_2a_3,b_1b_2b_3} \;,
\]
where:
\begin{align*}
    f_1^{\bA}   =    \frac{N-2}{ 2 N^{3/2}}  & \;, \qquad f_2^{\bA}  =  \frac{N^3 - 9 N^2 +32N-36 }{6^2   N^3 }  \;,  \crcr
    f_1^{\bS}   =   \frac{N^2+2N-8}{ 2 N^{3/2} (N+2)} & \;, \qquad f_2^{\bS}  =  \frac{N^6 +15 N^5 + 64N^4 -84N^3 -800N^2 +384 N +1536}{6^2 N^3 ( N+2)^3}  \;.
\end{align*}

In both cases $f_1^{\bP}$ is a series in $N^{- 1/2 }$ with $\lim_{N\to \infty } f_1^{\bP} = 0 $ while $f_2^{\bP}$ is a series in $N^{-1}$ with $  \lim_{N\to \infty} f_2^{\bP}  =  \frac{1}{36} $.
We denote by some abuse of notation $ \Sigma^{(2)  }=   \lambda K f_1^{\bP} +   \lambda^2 K^3 f_2^{\bP}$ and $T^4$ the interaction in Eq.~\eqref{eq:modelgen}. 

The subtracted interaction:
\[ \frac{ \lambda }{4 N^{3/2}} : T^4:_K =  \frac{ \lambda }{4 N^{3/2} } T^4 -  \frac{1}{2} \Sigma^{(2)} T \bP T  \;, \]
is ``Wick ordered'' up to second order in $\lambda$ with respect to the measure with covariance $K\bP$.
This interaction subtracts the tadpole and the melon contributions: the Feynman graphs of the model
with covariance  $K\bP$ and interaction $: T^4:_K $ have \emph{neither tadpoles, nor melon} subgraphs. 
It remains now to chose $K$ such that the model with covariance $K\bP$ and interaction subtracted with respect to $K$
is the original model of Eq.~\eqref{eq:crucial}:
\begin{align*}
  Z_{\bP}(\lambda) & = \left[  e^{\frac{1}{2} \frac{\partial}{\partial T}  \bP  \frac{\partial}{\partial T}  }  \; e^{  \frac{\lambda}{4 N^{3/2}} T^4}\right]_{T=0} = 
  \left[  e^{\frac{1}{2} \frac{\partial}{\partial T}  \bP  \frac{\partial}{\partial T}  }  \; e^{  \frac{   \Sigma^{(2)}  }{2}  T \bP T +  \frac{\lambda}{4 N^{3/2}} : T^4:_K }\right]_{T=0}  =
   \left[  e^{\frac{1}{2} \frac{1}{1 - \Sigma^{(2) }} \frac{\partial}{\partial T}  \bP \frac{\partial}{\partial T}  }  \; e^{  \frac{\lambda}{4 N^{3/2}} : T^4:_K }\right]_{T=0}   \;,
\end{align*}
therefore we chose $K$ such that: 
\[
 K = \frac{1}{  1 - \Sigma^{(2)}  } \Rightarrow   1 - K + \lambda f_1^{\bP}  K^2  +   \lambda^2 f_2^{\bP}  K^4  =0 \; . 
\]

For $N$ large and $\lambda$ small enough this equation admits a solution $K(\lambda,N)$ which is\footnote{We show in Appendix \ref{app:symmodel} that it is at this stage that the $1/N$ expansion 
fails for a symmetric tensor with no tracelessness condition: such a function $K(\lambda,N)$ does not exist in that case.}: a series in both $\lambda$ and $N^{-1/2}$, uniformly bounded in both $N$ and $\lambda$, such that
$\lim_{N\to \infty} K(\lambda,N)  $ is the generating function of the $4$-Catalan numbers, and
 \[ \lim_{\lambda \to 0} \left[  \lim_{N\to \infty} K(\lambda,N) \right] = 1 \;.\] 

It follows that Eq.~\eqref{eq:crucial} can be written as:
 \begin{align}\label{eq:crrrucial}
  F_{\bP}(\lambda) = \frac{4}{N^3} \lambda\partial_{\lambda} \ln \bigg\{
     \left[ e^{\frac{1}{2} \;   K (\lambda, N)  \partial_T \bP \partial_T   }      \; e^{  \frac{  \lambda  }{4 N^{3/2}}   
 :  \sum_{a_1\dots a_6}  T_{a_1a_2a_3}   T_{a_3a_4a_5}    T_{a_5 a_2 a_6}  T_{a_6a_4 a_1}:_{K(\lambda,N)} } \right]_{T=0}  
  \bigg\} \; ,
\end{align}
which is, as advertised, a theory with renormalized covariance $K (\lambda, N)$ and interaction which is both tadpole and melon subtracted
with respect to $K (\lambda, N)$. 
The perturbative expansion generates now Feynman graphs (which in turn expand in terms of stranded graphs) with \emph{no tadpoles and no melons}\footnote{The Feynman graphs $\hat \cS$ of 
 the subtracted theory are nothing but the melon-tadpole free cores discussed in  Section~\ref{sec:cores}.} and Eq.~\eqref{eq:perte1} becomes:
 \be\label{eq:rewrite}
 F_{\bP}(\lambda)  = \sum_{ \genfrac{}{}{0pt}{}{ \hat \cS  \text{ connected, rooted} }{  \text{with no tadpoles and no melons} } } \lambda^{ V (\hat \cS )} \left(  \frac{  \epsilon(\hat \cS )  }{3!^{U(\hat \cS ) + B(\hat \cS ) } } \; 
 \frac{ 1 }{ \left(  1 + \frac{2}{N} \right)^{ B(\hat \cS )} }  \right) \bigg[K(\lambda, N)  \bigg]^{U(\hat \cS ) + B(\hat \cS )}  
 \; \; N^{  - \omega(\hat \cS )} \;.
\ee

As $K(\lambda, N) $ is itself a series in $N^{-1/2}$, the $1/N$ expansion in Theorem~\ref{thm:main} follows from the following Proposition.
\begin{proposition}\label{prop:main}
 Let $\hat \cS $ be a connected stranded graph with no tadpoles and no melons. Then $\omega(\hat \cS) \ge 0$.
\end{proposition}
\begin{proof}
 Recall that the degree of $\hat \cS$ is defined as:
\[
\omega(\hat \cS) = 3 +  \frac{3}{2} V(\hat \cS)  + B(\hat \cS) - F(\hat \cS)  \;,
\]
where $V(\hat \cS), B(\hat \cS)$ and $F(\hat \cS)$ are the numbers of vertices, broken edges and faces of the map $\hat \cS$. 

We can replace all the broken edges in $\hat \cS$ by unbroken ones by cutting the two returning strands  
and regluing them the other way around into two traversing strands. This does not increase the degree: the number of faces can either increase or decrease by $1$, while the number of broken edges always 
decreases by $1$. Also, this does not introduce tadpoles or melons. 

Without loss of generality we restrict from now on to graphs with no tadpoles, no melons and  no broken edges, which we denote by $G$. 
All the strands go along the edges of $G$, hence the faces of length $q$ of $G$ are bounded by cycles of edges of length $q$. In particular the faces of length two are bounded by dipoles, and the faces of length three
by triangles. As $G$ has no tadpoles and no melons, $F_1(G)=0$. We iteratively apply the following:
\begin{itemize}
  \item either $G$ has dipoles, or it has no dipoles but has triangles. In virtue or Proposition~\ref{prop:moves} there exists a stranded graph $G'$ with no tadpoles, no melons, no broken edges,
        strictly fewer vertices than $G$, and degree not larger 
        than $G$,
 \item  or $G$ has no tadpoles, no dipoles and no triangles. Then $G$ is an end graph in the sense of Proposition~\ref{prop:end}, and in this case $\omega(G)\ge 0$ .
 \end{itemize}
 This proves the Proposition.
 
 In the next Section we explain how $G'$ is constructed from $G$ by erasing dipoles and triangles, following a precise algorithm. This construction is quite subtle and 
 in particular it is not true that $G'$ has, for instance, fewer dipoles than $G$ (the deletion of a dipole can in fact create two dipoles). However, our induction is 
 on the number of vertices of the graph, hence ends in a finite number of steps.
 \end{proof}

\newpage

\section{Deletions}
\label{sec:proof1}

From now on graphs are understood to have no broken edges.

\subsection{Dipoles}
 
\begin{lemma}\label{lem:melon}
 Let $G$ be a graph with only unbroken edges which has a melon. Then there exists a graph $G'$ obtained from $G$ by replacing the melon with an unbroken edge such that $\omega(G')\le \omega(G)$. 
\end{lemma}
\begin{proof}
 There are 9 internal strands provided by the edges. The internal face have even length. An external face of the melon
 \begin{itemize}
  \item[--] either traverses from one external halfedge to the other, in which case it has odd length;
  \item[--] or loops back and returns on the same external halfedge, in which case it has even length.
 \end{itemize}

 If all the three external faces of the melon traverse, there are at most 6 remaining strands which can close into at most 3 internal faces. Deleting the vertices and the internal faces does not increase the degree.
 
 If two external faces of the melon loop back, then there are at most 5 strands left to support internal faces, and therefore at most two internal faces. We delete the two vertices and the two internal faces.
 This creates a broken edge. As before, we can replace the broken edge by an unbroken one, which can at most delete another face. In total two vertices  are deleted and 
 at most three faces are also deleted. Thus the degree can not increase. 
 
\begin{figure}[htb]
 \begin{center}
 \includegraphics[scale=.3]{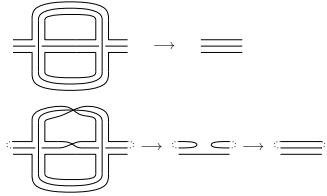}  
 \caption{Examples of deletions of melons which do not change the degree. The dots represent a specific configuration of faces
 dictated by the rest of the graph. The three external strands traverse the melon in the top panel, while two external strands loop 
 back in the bottom one; in both cases, at most 3 faces are deleted by deleting the melon.}\label{fig:deletemelons}
 \end{center}
 \end{figure}
 
  Examples of limiting cases in which the degree is unchanged after the deletion are provided in Fig.~\ref{fig:deletemelons}. 
  
\end{proof}

 \begin{figure}[htb]
 \begin{center}
 \includegraphics[width=10cm]{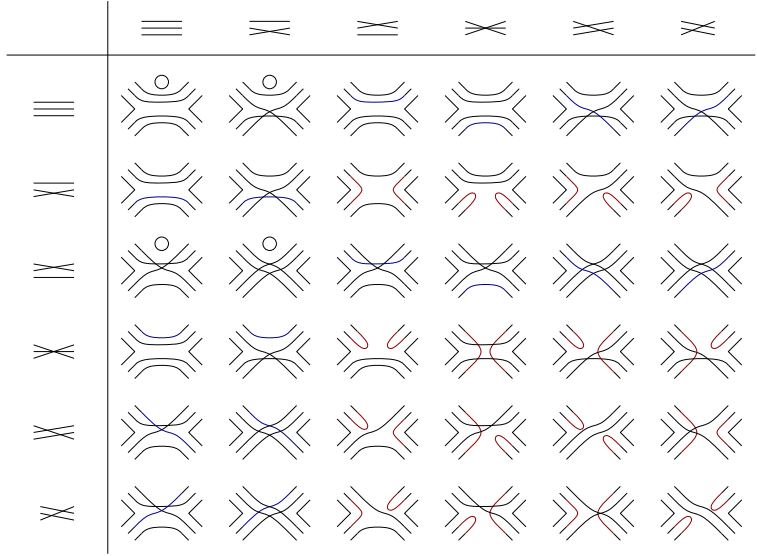}  
 \caption{The faces of a 2-dipole with unbroken edges. There are $6$ available 
 strands (3 for each edge). We colored in black the strands of length $1$, in red those of length $2$ and in blue those of length $3$. The
 odd length strands go from left to right and the even length ones 
 leave and return to the same side. Four of the cases have an internal face of length $2$ which we draw as a small circle. } \label{fig:36}
 \end{center}
 \end{figure}

 A face of length two is bounded by a $2$--dipole, but not every $2$--dipole encloses a face of length $2$. 
All the possible configurations of connexions of faces for a $2$--dipole are presented in Fig.~\ref{fig:36}.

  \begin{figure}[htb]
 \begin{center}
 \includegraphics[scale=.6]{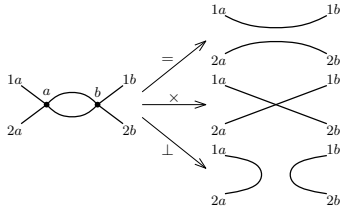}  
 \caption{Deletion of dipoles. The vertices $a$ and $b$ are deleted and the half edges are reconnected in a pairing. Some pairs of internal strands 
 need to be cut and reglued.} \label{fig:channels}
 \end{center}
 \end{figure}

 A 2--dipole can be \emph{deleted}. This consists in deleting the two vertices and reconnecting the strands into two new edges. 
 There are a priori three channels of deletion \emph{parallel} $=$, \emph{cross} $\times$ and \emph{orthogonal} $\perp$ (with respect to the parallel edges of the dipole)
 represented in Fig.~\ref{fig:channels}. We denote $G'$ the graph obtained from $G$ by a deletion of a 
 $2$--dipole. The deletion always erases two vertices, $V(G') = V(G)-2$. Concerning the variation of the number of faces:
\begin{itemize}
 \item[--] the internal face of length two, if it exists, is erased.
 \item[--] pairs of external strands are \emph{cut and reglued}. Each cut-and-glue operation:
  \begin{itemize}
   \item either deletes a face, if two different faces of $G$ are cut and reglued in a face of $G'$;
   \item or creates a face or leaves the number of faces unchanged (depending on the structure of the face and the regluing), if the same face of $G$ is cut twice and reglued.
  \end{itemize}
\end{itemize}

Observe that a deletion can separate  connected components, hence can change the number of connected components of $G$. 

\begin{lemma} 
 Any dipole can be deleted in the orthogonal channel $\perp$ in such a way that:
 \[
  F(G') \ge F(G)-3 \;.
 \]
 Any dipole can be deleted in either the parallel $=$ or the cross $\times$ channel (or both) in such a way that:
 \[
  F(G') \ge F(G)-3 \;.
 \]
 If $G'$ is connected then $\omega(G')\le \omega(G)$.
\end{lemma}
\begin{proof}
The first two statements follow by direct inspection of the 36 cases in Fig.~\ref{fig:36}. 

We consider the deletion in the orthogonal channel. In the 32 cases in which the dipole does not have an internal face of length 2,
one can always perform the deletion by cutting a maximum of $4$ strands (including the red strands which sometimes need to be cut in order to prevent the creation of broken edges).
In the worst case, the cut-and-glue operations will delete 3 faces by merging 4 distinct faces into one.
 \begin{figure}[htb]
 \begin{center}
 \includegraphics[scale=.7]{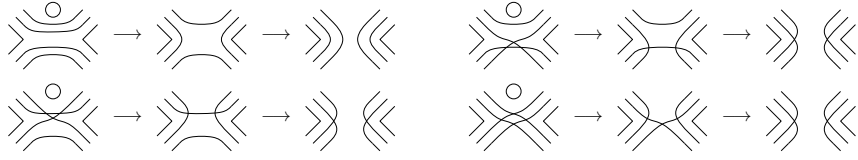}  
 \caption{$\perp$ deletions by $2$ successive cut-and-glue operations.} \label{fig:cut-glue_dipole1}
 \end{center}
 \end{figure}
If the dipole has an internal face of length $2$ (the 4 cases in Fig.~\ref{fig:36}), one can always perform the deletion with no more than two cut-and-glue operations
as illustrated in Fig.~\ref{fig:cut-glue_dipole1}. In all cases $F(G') \ge F(G) - 3$. 
 
In the parallel and cross channels, the number of strands to be cut ranges from 2 to 6. For all cases in which it is 4 or fewer in the parallel channel,
the deletion can never affect more than four faces, and therefore $F$ cannot decrease by more than 3. When the number of strands one needs to cut in the parallel channel is 6, 
it is 4 or fewer in the cross channel, so one can delete in the cross channel. 

 \begin{figure}[htb]
 \begin{center}
 \includegraphics[scale=.8]{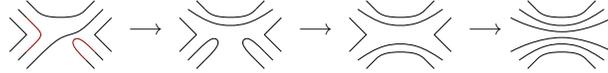}  
 \caption{Example of $2$-dipole deletion requiring five cuts in the parallel channel. 
The deletion can be implemented with $3$ successive cut-and-glue operations.} \label{fig:contract5cuts}
 \end{center}
 \end{figure}

One is left with 8 cases, which require 5 cuts in either of the $=$ or $\times$ channels.
Such dipoles do not have an internal face, and one can check that $3$ successive cut-and-glue operations suffice in all cases.
An example is provided in Fig.~\ref{fig:contract5cuts}.
Again, in all cases $F(G') \ge F(G) - 3$. 

The last statement is immediate.

\end{proof}

\begin{lemma}\label{lem:2pointmoves}
Consider the four $2$-point graphs $H_1,H_2,H_3,H_4$ of Fig.~\ref{fig:adegchains}. In a graph with only unbroken edges, replacing any subgraph $H_i$ by an unbroken edge strictly 
lowers the degree. We call any such combinatorial move a \emph{$H$-contraction}. More precisely:
\begin{itemize}
\item $\omega \to \omega' \leq \omega -1$ under a $H_1$-contraction;
\item $\omega \to \omega' \leq \omega - 1/2$ under a $H_2$-contraction;
\item $\omega \to \omega' \leq \omega - 1$ under a $H_3$-contraction;
\item $\omega \to \omega' \leq \omega - 1$ under a $H_4$-contraction.
\end{itemize}
\end{lemma}
\begin{proof}
The proof is lengthy but straightforward. Details can be found in Appendix~\ref{app:Specialcases}.

\end{proof}

\begin{lemma}\label{lem:dipole}
Let us consider a graph $G$ with no melon-tadpoles. If $G$ has a $2$-dipole, then there exists a graph $G'$ with strictly fewer vertices than $G$ having no tadpoles and no melons and
with $\omega(G') \le \omega(G)$.
\end{lemma}
\begin{proof} As $G$ has no tadpole and no melon, no two external halfedges of the dipole can be connected into an edge.
We first search for dipoles which can be easily deleted in the orthogonal channel $\perp$. There are several cases:
\begin{description}
 \item[There exists a dipole whose deletion in the $\perp$ channel disconnects the graph.] We then build $G'$ by deleting this dipole in the parallel channel $=$ (or cross channel $\times$). 
 This deletion cannot disconnect the graph, cannot increase the degree and cannot create tadpoles or melons.
 \item[All deletions in the $\perp$ channel do not disconnect the graph.] We then have two sub cases:
 \begin{itemize}
  \item{\it There exists a deletion $\perp$ which does not create a tadpole or a melon.} We then obtain $G'$ by performing the deletion.
    \item{\it All deletions $\perp$ create at least a tadpole or a melon.} Any dipole is then in one of the configurations depicted in Fig.~\ref{fig:adelchains}. 
    In all the three cases the deletion in the parallel $=$ (or cross $\times$) channel can not disconnect the graph.
 \begin{figure}[htb]
 \begin{center}
 \includegraphics[scale=.6]{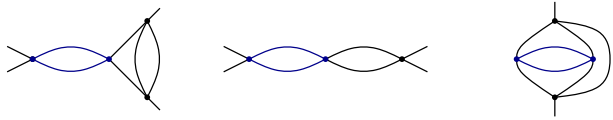}  
 \caption{The deletion of the (blue) dipole in the channel $\perp$ creates at least a melon or a tadpole.} \label{fig:adelchains}
 \end{center}
 \end{figure}
 Observe also that, in the two $4$-point drawings, the two right (resp. left) halfedges of the chain cannot be connected into an edge (as $G$ has no tadpoles and no melons), but a left and a right ones can. We then have two sub cases:
 \begin{itemize}
    \item there exists a deletion in  the parallel $=$ (or cross $\times$) channel which does not create a tadpole or a melon. We then build $G'$ by deleting in the parallel $=$ (or cross $\times$) channel.
   \item all deletions in the  parallel $=$ (or cross $\times$) channel create a tadpole or a melon. Taking into account that the two left (resp. two right) halfedges can not be paired together, we are
   then in one of the four cases depicted in Fig.~\ref{fig:adegchains}.
  \begin{figure}[htb]
 \begin{center}
 \includegraphics[scale=.6]{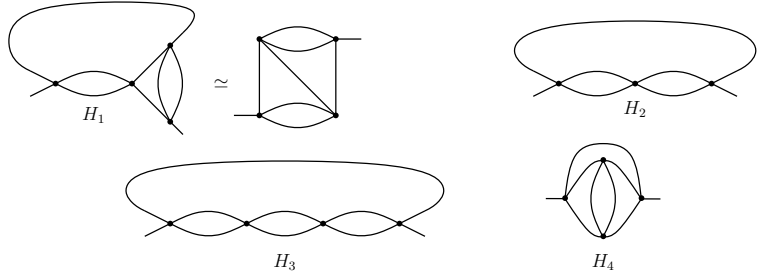}  
 \caption{The cases in which both the $\perp$ and the $=$ (or $\times$) create at least a tadpole or a melon.} \label{fig:adegchains}
 \end{center}
 \end{figure}
 
We now proceed with a $H$-contraction shown in Fig.~\ref{fig:excludedH}, which, by lemma \ref{lem:2pointmoves}, decreases the  degree. 
\begin{figure}[htb]
 \begin{center}
 \includegraphics[scale=.6]{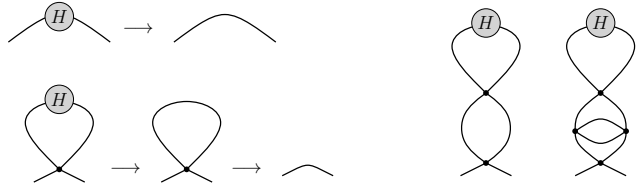}  
 \caption{$H$-contractions. On the left panel are represented the two allowed situations: either a tadpole is not generated, or one is generated and deleted. 
 The chains on the right panel have already been excluded.}\label{fig:excludedH}
 \end{center}
 \end{figure} 

This cannot create a melon\footnote{Otherwise there exists a dipole in $G$ which, together with the edge replacing $H_i$, becomes the melon. This dipole could 
itself be deleted in the $\perp$ channel without generating a melon or a tadpole, which we already excluded.}. 
If no tadpole is created, we have our graph $G'$. If a tadpole is created, we obtain $G'$ after replacing it by an unbroken edge. 
The degree can not increase by more that $1/2$ when deleting the tadpole, but this is compensated by the net decrease of the degree from the $H$-contraction.
$G'$ can have neither a tadpole nor a melon\footnote{Otherwise we are in the situation on the right panel in Fig.~\ref{fig:excludedH} which has already been 
excluded.}.
 
 \end{itemize}
 
 \end{itemize}

\end{description}

\end{proof}

 \newpage
 
 \subsection{Triangles}\label{sec:triangles}
 
In order to facilitate the description of the various strand configurations of triangles, we will use 
the notion of \emph{boundary graph}, which we now introduce. Consider a $n$-point graph $H$, with external legs $l_1 , \ldots , l_n$. Its boundary 
$\partial H$ is the closed graph consisting of $n$ $3$-valent vertices $v_1, \,\ldots,\, v_n$ such that: to every strand connecting $l_i$ to $l_j$ in $H$ corresponds an edge 
of $\partial H$, connecting $v_i$ to $v_j$. Examples of dipoles with their boundary graphs are presented in Fig.~\ref{fig:boundary-ex}. The boundary graphs 
can immediately be obtained from the stranded representations of Fig.~\ref{fig:36}, by: a) deleting any closed face; b) pinching each triplet of open strands to form a $3$-valent vertex. 
\begin{figure}[htb]
 \begin{center}
 \includegraphics[scale=.7]{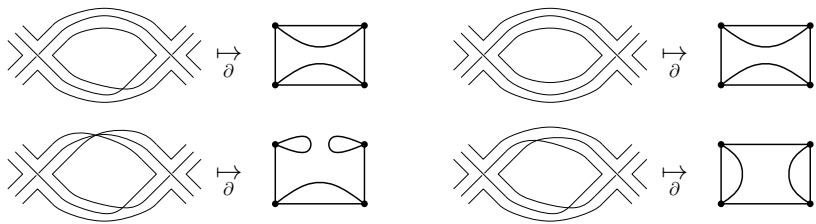}  
 \caption{Some dipoles and their boundary graphs.} \label{fig:boundary-ex}
 \end{center}
 \end{figure}
 
The boundary graph tracks the external strands of $H$. Its key property is the following: the rest of the graph $G\setminus H$ ``sees'' only the boundary $\partial H$ of a subgraph $H$.
To be precise, consider two open graphs $H$ and $\tilde{H}$, such that $\partial H = \partial \tilde{H}$. Given a graph $G$ having $H$ as a subgraph, 
one can construct a graph $\tilde{G}$ by replacing $H$ with $\tilde{H}$ and gluing the external strands of $\tilde H$ according to the 
pattern of gluing of the external strands of $H$. 
The graph $\tilde G$ is such that: $V(\tilde{G}) =  V(G) - V(H) + V(\tilde{H})$ and $F(\tilde{G}) =  F(G) - F(H) + F(\tilde{H})$.

In particular, we will often look to replace triangles in $G$ with pairs of vertices connected by one (unbroken) edge. The boundary graph of 
two vertices connected by one unbroken edge is always the \emph{prism graph} shown in Fig.~\ref{fig:boundary2}. It will play an important role in the following.
\begin{figure}[htb]
 \begin{center}
 \includegraphics[scale=.7]{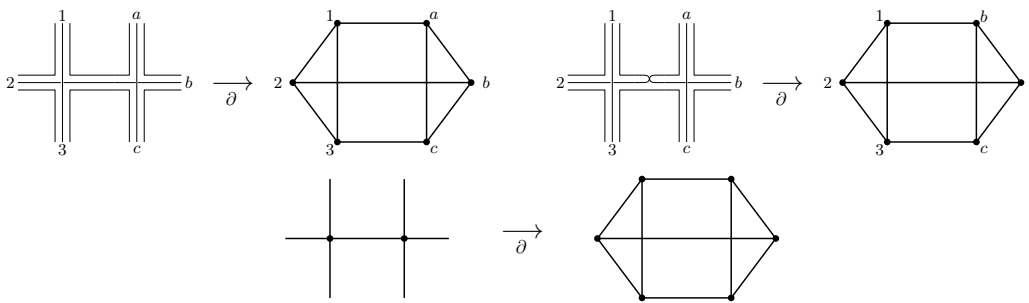}  
 \caption{Top panel: two examples of $6$-point diagrams with two vertices connected by one unbroken edge. Bottom panel: the boundary graph of such diagrams is always a prism.} \label{fig:boundary2}
 \end{center}
 \end{figure}

 \subsubsection{Boundary graphs of triangles}
 
 Let us consider a triangle $H$ and let us label its vertices $a$, $b$ and $c$ and its external half edges $a_1$ and $a_2$ (incident at $a$), 
 $b_1$ and $b_2$ (incident at $b$) and  $c_1$ and $c_2$ (incident at $c$). $H$ has $18$ corners. We call a 
 corner \emph{external} (resp. \emph{internal}) if it separates a pair of external half-edges (resp. internal edges), and \emph{mixed} otherwise. 
There are $3$ external (resp. internal) corners ($1$ per vertex of $H$), and $12$ mixed corners ($4$ per vertex of $H$).

The boundary graph $\partial H$ has 6 vertices: $a_1,a_2,b_1,b_2,c_1,c_2$ and 9 edges. 
We color the edges of $\partial H$ according to their length\footnote{The numbers of corners minus $1$ traversed by the corresponding strand in $H$.}: 
gray for length $0$, black for $1$, red for $2$, blue for $3$, and green for $4$. 
The total available length is 9. The vertices of $\partial H$ are partitioned into 3 pairs of vertices connected 
by gray edges\footnote{They represent the pairs of half edges of the triangle incident to the same vertex.}: 
$(a_1,a_2)$, $(b_1,b_2)$ and $(c_1,c_2)$, as depicted in Fig.~\ref{fig:boundarystart}. All the possible boundary graphs are obtained by adding 6 more edges to build a trivalent graph.
\begin{figure}[htb]
 \begin{center}
 \includegraphics[scale=.7]{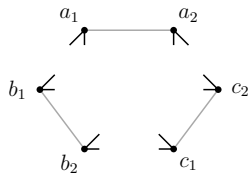}  
 \caption{The starting configuration for building $\partial H$} \label{fig:boundarystart}
 \end{center}
 \end{figure}

 In order to list all the possible boundary graphs observe that:
 \begin{itemize}
  \item the two strands originating at a half edge, say $a_1$, of $H$ point one towards the vertex $b$ and one towards the vertex $c$. Consider the strand pointing towards $b$. It 
 can either exit through one of the halfedges of $b$, or it can go through the internal corner of $b$ and head on towards $c$. 
 We will sometimes label the strands (which become edges in $\partial H$) by their end half edges (vertices of $\partial H$) and by the ordered list of internal vertices they go through.
 \item the three strands of the internal edge $(a,b)$ of $H$ originate one in $a_1$, another one in $a_2$ and the third one on the internal corner of $a$.
 \end{itemize}
 
\

\paragraph{Triangles with an internal face.} 
If a closed face is bounded by the triangle, it must necessarily be of length $3$ and go through the $3$ internal corners. 
The 6 edges of $\partial H$ must all have length 1 (hence are black edges). Recalling the structure of the stranded vertex we 
have:
\begin{itemize}
 \item[--] one edge incident to $a_1$ must connect to one of the $b$ vertices (say $b_1$) and the other edge incident to $a_1$ to one of the $c$ vertices (say $c_1$).
 \item[--] one edge incident to $a_2$ must connect to the other $b$ vertex (that is $b_2$) and the other edge incident to $a_2$ to the other $c$ vertex (that is $c_2$). 
 \item[--] the remaining two edges can be: 
    \begin{itemize}
     \item either $(b_1,c_2)$ and $(b_2,c_1)$, leading to the complete bipartite graph $\cG_1$ in Fig.~\ref{fig:triangle_face},
     \item or $(b_1,c_1)$ and $(b_2,c_2)$, leading to the prism graph $\cG_2$ in Fig.~\ref{fig:triangle_face}.
    \end{itemize}
\end{itemize} 

\begin{figure}[htb]
 \begin{center}
 \includegraphics[scale=.7]{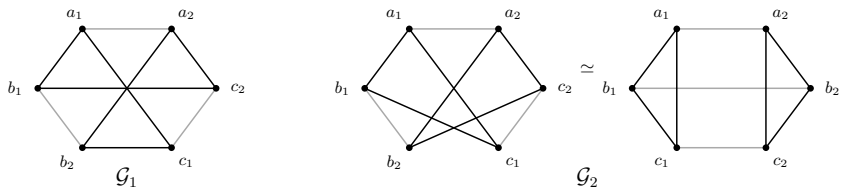}  
 \caption{The boundary graphs of triangles with an internal face.} \label{fig:triangle_face}
 \end{center}
 \end{figure}

 One can easily reconstruct the corresponding triangles. They can always be embedded in the plane as shown in Fig.~\ref{fig:cases3} at the top. 
 In the remainder of this paper, we will call \emph{untwisted triangle} (resp. \emph{twisted triangle}) any triangle which can be brought in the canonical form shown on the left (resp. right).
\begin{figure}[htb]
 \begin{center}
 \includegraphics[scale=.5]{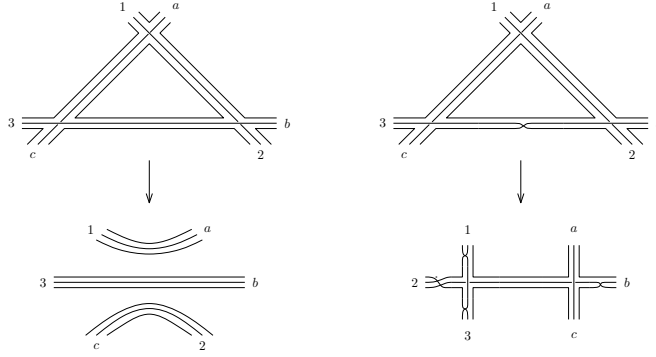}  
 \caption{Canonical forms of untwisted and twisted triangles.} \label{fig:cases3}
 \end{center}
 \end{figure}

\paragraph{Triangles with no internal face.} When no closed face is bounded by the three edges of the triangles, there are more cases to distinguish:
\begin{description}
 \item[\it $(1+1+1)$-triangle.]
We call \emph{$(1+1+1)$-triangle} a triangle in which the three internal corners belong to three distinct strands.
$\partial H$ will have $3$ edges of length $2$ (red edges), and 3 edges of length $1$ (black edges).
The red and black edges can not connect vertices in the same pair $a$, $b$ or $c$. Furthermore, 
every pair $a$, $b$ or $c$ is incident to exactly two black and two red edges\footnote{Because along any edge of $H$, exactly one strand is black and two are red.}.
\begin{figure}[htb]
 \begin{center}
 \includegraphics[scale=.7]{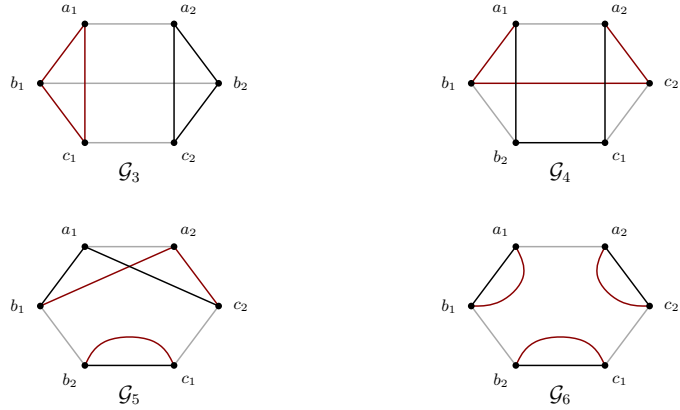}  
 \caption{The boundary graphs of the $(1+1+1)$-triangles. $\cG_3$ and $\cG_4$ are prisms.} \label{fig:triangle_1-1-1}
 \end{center}
 \end{figure}
We have several cases:
\begin{itemize}
 \item[--] on every pair, the red edges are incident to the same vertex. We can consider the red edges to be $(a_1,b_1)$, $(b_1,c_1)$ and $(a_1,c_1)$.
 The black edges must then be $(a_2,b_2),(b_2,c_2)$ and $(a_2,c_2)$ and we find the case $\cG_3$ in Fig.~\ref{fig:triangle_1-1-1}.
 \item[--] on exactly two pairs the red edges are incident to the same vertex. Without loss of generality we can assume that the red edges are $(b_1,a_1)$, $(b_1,c_2)$ and $(c_2,a_2)$. 
  Then the black edge incident to $a_1$ must also connect to the $b$ pair\footnote{The red strand corresponding to $(a_1,b_1)$ is $a_1cb_1$, hence the 
  last strand exiting from $a_1$ points to $b$.}, 
  but it can only connect on $b_2$, hence we get the edge $(a_1,b_2)$. Similarly we get a black edge $(a_2,c_1)$. Finally, the last black edge must be 
  $(b_2,c_1)$ leading to the graph $\cG_4$ in Fig.~\ref{fig:triangle_1-1-1}.
 \item[--] on only one pair the red edges are incident to the same vertex. Without loss of generality we can assume that $(a_2,b_1)$ and $(a_2,c_2)$ are red. But then there must exist a red edge $(b_2,c_1)$.
 Now, there exists a black edge incident to $c_2$ which must connect to the pair $a$, and it can only connect to $a_1$: hence there is an edge $(c_2,a_1)$. Similarly we obtain the black edge $(b_1,a_1)$, and finally another
 black edge $(b_2,c_1)$, leading to $\cG_5$ in Fig.~\ref{fig:triangle_1-1-1}.
 \item[--]  finally, on no pair the two red edges are incident to the same  vertex.  Without loss of generality we can assume that $(a_1,b_1)$, $(a_2,c_2)$ and $(b_2,c_1)$ are red edges. But then
  the black edge coming out of $a_1$ must connect to the pair $b$\footnote{As $(a_1,b_1)$ is the strand $a_1cb_1$.}, and it can only connect to $b_1$\footnote{Because, by the same argument, the black edge coming out of $b_2$ must connect on the pair $c$.}. 
  We obtain $\cG_6$ in Fig.~\ref{fig:triangle_1-1-1}.
\end{itemize}
 \item[\it $(2+1)$-triangle.]
We call \emph{$(2+1)$-triangles} the triangles in which exactly two of the internal corners belong to the same strand.
$\partial H$ will have an edge of length 3 (blue), one of length $2$ (red) and four of length $1$ (black).
The blue edge returns on the same pair (say $a$), and the red edge must connect the two other pairs ($b$ and $c$).
 \begin{figure}[htb]
 \begin{center}
 \includegraphics[scale=.7]{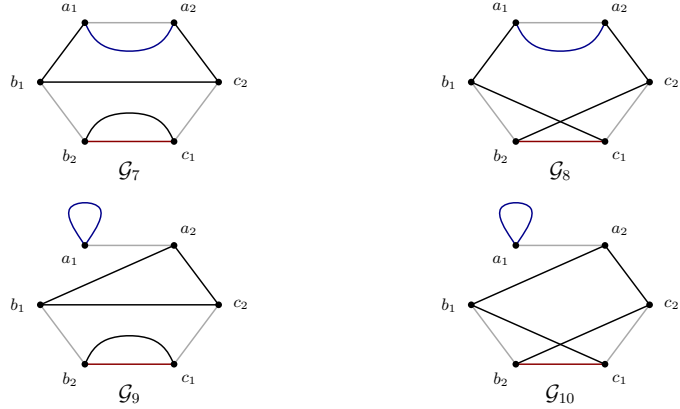}  
 \caption{The boundary graphs of the $(2+1)$-triangles.} \label{fig:triangle_2-1}
 \end{center}
 \end{figure}
We have several cases: 
\begin{itemize}
 \item[--] the blue edge connects $a_1$ and $a_2$. Without loss of generality we can assume that the corresponding strand in $H$ is $a_1cba_2$. 
 Then there exists a black edge in $\partial H$ starting in $a_1$ and ending on the pair $b$ (and one starting in $a_2$ and ending on the pair $c$). 
 Without loss of generality we can label them $(a_1,b_1)$ and $(a_2,c_2)$. As the red strand must pass through the internal corner on the vertex $a$,
 it must be $b_2 a c_1$, hence $\partial H$ has a red edge $(b_2,c_1)$. We have two sub cases:
 \begin{itemize}
  \item either the two remaining black edges are $(b_1, c_2)$ and $(b_2, c_1)$, leading to $\cG_7$ of Fig.~\ref{fig:triangle_2-1}.     
  \item or the two remaining black edges are $(b_1, c_1)$ and $(b_2, c_2)$, leading to $\cG_8$ of Fig.~\ref{fig:triangle_2-1}.
 \end{itemize}
 \item[--] the blue edge goes from $a_1$ to $a_1$.  Without loss of generality we can assume that the corresponding strand in $H$ is $a_1cba_2$. 
 Without loss of generality we can assume that $(a_2,b_1)$ and $(a_2,b_2)$ are black edges. Again the red edge must pass through the internal corner on the vertex $a$,
 hence it must be $b_2 a c_1$, and $\partial H$ has a red edge $(b_2,c_1)$. We again get two sub cases:
   \begin{itemize}
     \item either $(b_1, c_2)$ and $(b_2, c_1)$ are the last two black edges, and we get $\cG_9$ of Fig.~\ref{fig:triangle_2-1}.
     \item or the last two black edges are $(b_1, c_1)$ and $(b_2, c_2)$, yielding $\cG_{10}$ in Fig.~\ref{fig:triangle_2-1}.
   \end{itemize}

\end{itemize}
 
\item[\it $3$-triangle.]
Finally, we call \emph{$3$-triangle} a triangle in which the three internal corners belong to the same external strand. 
The latter is then necessarily of length $4$ and goes twice through the same edge of the triangle. As far as the boundary graph is concerned, this is 
identical to the case of a triangle with an internal face and we obtain the same two boundary graphs, as shown in Fig.~\ref{fig:triangle_3}.
\begin{figure}[htb]
 \begin{center}
 \includegraphics[scale=.7]{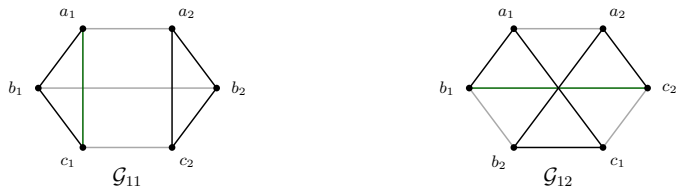}  
 \caption{The boundary graphs of $3$-triangles.} \label{fig:triangle_3}
 \end{center}
 \end{figure}

 \end{description}

\subsubsection{Deletion of triangles}

We call \emph{$3 \to 2$ move} any local combinatorial move which consists in replacing a triangle by two vertices connected by one edge, as shown in Fig.~\ref{fig:3to2}. A $3 \to 2$ move reduces the number of vertices by $1$ 
and cannot create new connected components. This operation can be used to eliminate a triangle in a way which minimizes the non-local effects on the face structure of a graph.  
\begin{figure}[htb]
 \begin{center}
 \includegraphics[scale=.6]{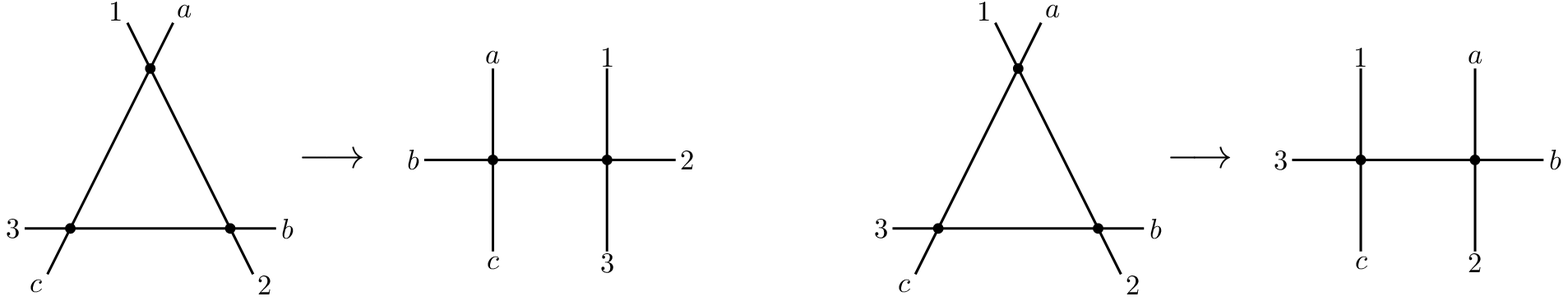}  
 \caption{The two types of $3\to2$ moves.} \label{fig:3to2}
 \end{center}
 \end{figure}

 There are two types of $3 \to 2$ moves. The \emph{left} $3 \to_L 2$ move, represented on the left in Fig.~\ref{fig:3to2} is such that 
 any two half-edges initially hooked to the same vertex end up hooked to different vertices. The  \emph{right} $3 \to_R 2$ move, represented on the right in Fig.~\ref{fig:3to2} is such that 
 an edge of the initial graph is conserved.
 
\begin{figure}[htb]
 \begin{center}
 \includegraphics[scale=.5]{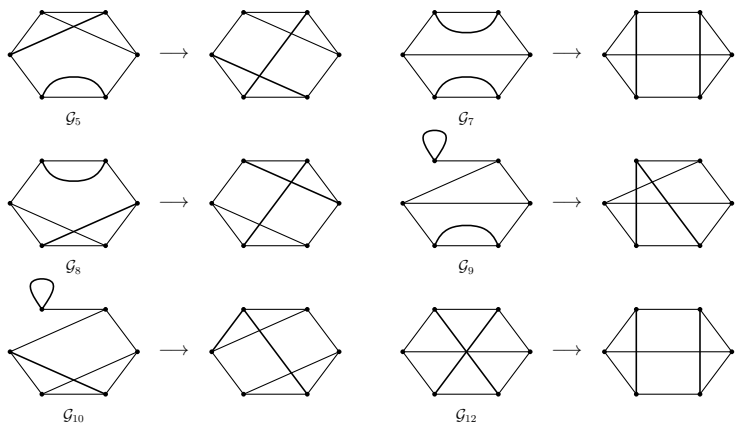}  
 \caption{Cut-and-glue operations turning boundary graphs of triangles into prisms (the lines being cut and glued are represented with a heavier stroke).} \label{fig:cut-glue_triangles}
 \end{center}
 \end{figure}

\begin{lemma}\label{lem:3to2_moves}
Let $G$ be a graph with a triangle $H$. Assume that $H$ is neither untwisted, nor of type $(1+1+1)$ with boundary $\cG_6$. Then $H$ can be eliminated by a $3\to 2$ move, such that the resulting graph $G'$ has degree:
\be
\omega(G') \leq \omega(G) - 1/2\;.
\ee
\end{lemma}
\begin{proof}
If $H$ is twisted, then one performs the $3\to 2$ move depicted in the right panel of Fig.~\ref{fig:cases3}. We have $V(G') = V(G)-1$ and $F(G') = F(G)-1$, hence $\omega(G') = \omega(G) -\frac{1}{2}$.

Similarly, one can perform a $3\to 2$ which preserves the external face structure for any triangle whose boundary graph is a prism. This is the case of triangles with boundary graphs $\cG_3$, $\cG_4$, and $\cG_{11}$.
Such triangles have no internal face, so that this time $F(G')=F(G)$ and therefore $\omega(G') = \omega(G) - 3/2  < \omega(G) - 1/2$.  

For triangles with boundaries $\cG_5$, $\cG_7$, $\cG_8$, $\cG_9$, $\cG_{10}$ and $\cG_{12}$, one first performs a cut-and-glue operation, as shown in Fig.~\ref{fig:cut-glue_triangles}. This reduces 
the number of faces by at most $1$, and yields a prism boundary graph. We can subsequently perform a $3\to2$ move to obtain $G'$ with $\omega(G') \leq \omega(G) - 1/2$.

\end{proof}

\

Later in the proof, it will be convenient to delete two triangles simultaneously if they are \emph{adjacent} to each other, that is if they share exactly one edge.  
\begin{lemma}\label{lem:4to2}
Let $G$ be a graph with no tadpole, no melon, and no dipole. Suppose $G$ contains a pair of adjacent untwisted triangles. Then there exists a  graph $G'$ having strictly fewer vertices, having no melons, no tadpoles and with $\omega(G') \le \omega(G)$.
\end{lemma}
\begin{proof}
Let us call $H$ and $H'$ the two adjacent triangles. Being untwisted, they must be in the configuration 
shown on the left side of Fig.~\ref{fig:2twisted} (up to some permutations of the external half-edges)\footnote{Consider $H$ and $H'$ represented on the left-hand side of Fig.~\ref{fig:2twisted}. 
We now add the strands on all the edges.
The edge common to $H$ and $H'$ must have parallel strands (otherwise it could not close the two internal faces of the triangles).
Consider the upper triangle. Both edges have one strand fixed by the requirement that the triangle closes a face. If the other two 
strands on only one of the edges are twisted, then the upper triangle is twisted. If both are twisted, then they can be both straightened by a 
permutation of the upper halfedges. The same applies to the lower triangle yielding the canonical form on the left panel of Fig.~\ref{fig:2twisted}.}. 
Their boundary graph being a prism, they can be
eliminated simultaneously by a $4\to 2$ move preserving the external strand structure (see Fig.~\ref{fig:2twisted}). Two internal
faces and two vertices are deleted in the process and the resulting graph $\tilde{G}$ has degree:
\be
\omega(\tilde{G}) = \omega(G) - 2 \times \frac{3}{2} + 2 = \omega(G) - 1\;.
\ee
\begin{figure}[htb]
 \begin{center}
 \includegraphics[scale=.6]{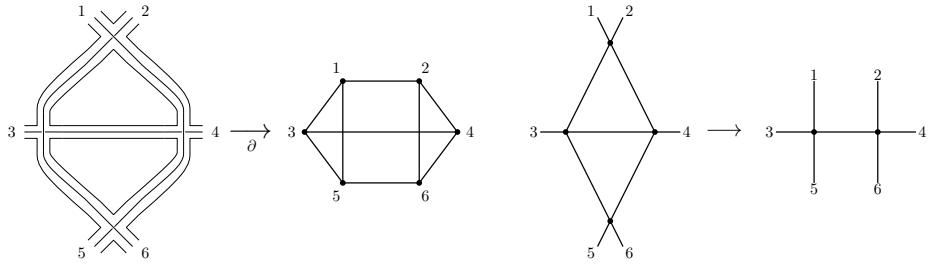}  
 \caption{Two adjacent untwisted triangles, their boundary prism graph and the $4 \to 2$ move removing them.} \label{fig:2twisted}
 \end{center}
 \end{figure}
\begin{description}
\item[\it $\tilde G$ has no tadpole and no melon.] If $\tilde{G}$ has a no tadpole and no melon, we take $G' = \tilde{G}$. 
\item[\it $\tilde G$ has a tadpole.] As $G$ has no dipoles and no tadpoles, the only way for $\tilde G$ to have a tadpole is for the half edges $1$ and $5$ (or 2 and 6) to be connected into an edge.
Only one of these two pairs can form an edge, otherwise $G$ would have a dipole. Suppose that $(1,5)$ is an edge of $G$. 
Then one can delete the tadpole in $\tilde G$ to obtain a graph $\hat G$ with $\omega(\hat G) \leq \omega(\tilde G) + 1/2 = \omega(G) - 1/2$. The new graph $\hat G$ cannot have tadpoles 
(this would require that $3$ and $4$ are connected in $G$ which would form a dipole), but it could have a melon. 
Deleting it, as well as at most another tadpole obtained after the melon deletion, we obtain $G'$ with $\omega(G') \leq \omega(\hat G) + 1/2 = \omega(G)$.
We have represented the worst case in Fig.~\ref{fig:2twisted_tad}.

\begin{figure}[htb]
 \begin{center}
 \includegraphics[scale=.4]{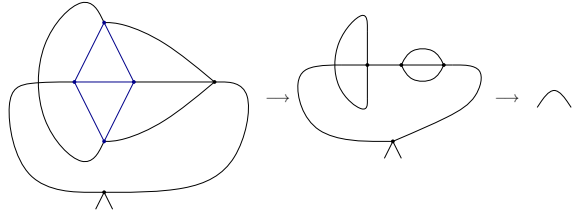}  
 \caption{A $4\to2$ move which generates a tadpole and a melon, nested into a tadpole.} \label{fig:2twisted_tad}
 \end{center}
 \end{figure}
\item[\it $\tilde G$ has no tadpole, but has a melon.]
If $\tilde G$ has no tadpole but has melons, the latter can be eliminated. This operation may generate tadpoles, but this can only happen for the
cases shown in Fig.~\ref{fig:2twisted_mel}. The two graphs on the left can generate at most one tadpole, while the two diagrams on the right may generate up to two tadpoles. In both cases, one
may delete them and obtain a graph $G'$ with $\omega(G') \leq \omega(\tilde G) + 1 = \omega(G)$. 
\begin{figure}[htb]
 \begin{center}
 \includegraphics[scale=.4]{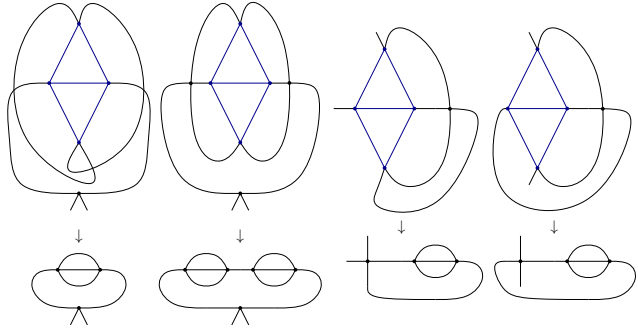}  
 \caption{$4\to2$ move which generates melons and then tadpoles after eliminating the melons.} \label{fig:2twisted_mel}
 \end{center}
 \end{figure}  
\end{description}
\end{proof}

\

In view of Lemma \ref{lem:dipole}, we consider only graphs $G$ having no melon-tadpoles no $2$-dipoles and, of course, no broken edges.
\begin{lemma}\label{lem:tripole}
Let $G$ be a graph with only unbroken edges having no melons, no tadpoles and no $2$-dipoles. If $G$ has a triangle, then there exists a (possibly 
disconnected\footnote{Recall that the degree of a disconnected graph is the sum of the degrees of its connected components: $\omega(G) = 3 C(G) + 3/2 V(G) - F(G)$, where 
$C(G)$ is the number of connected components of $G$.}) graph $G'$ having strictly fewer vertices, having 
no melons, no tadpoles and with $\omega(G') \le \omega(G)$.
\end{lemma}

\begin{proof}
As $G$ has no tadpoles, no melons and no dipoles, no pair of half edges incident to the triangle can be glued together into an edge. We discuss a number of cases separately.

\begin{description}
\item[There exists a triangle $H$ which is neither untwisted, nor with boundary graph $\cG_6$.] From Lemma~\ref{lem:3to2_moves}, one can perform a $3\to 2$ move to remove $H$ 
and obtain a graph $\tilde G$ with degree $\omega(\tilde G) \leq \omega(G) - 1/2$. Observe that for both left and  right moves, $3\to_L 2$ and $3\to_R 2$ 
$\tilde G$ can not have tadpoles\footnote{Otherwise a pair of external half edges of the triangle are joined into an edge, which is impossible.}.
We have two cases:
\begin{itemize}  
\item{\it $3\to_L 2$ move.} If $\tilde G$ has no melon then we set $G' = \tilde G$. Assume that $ \tilde G$ has a melon. 
 \begin{figure}[htb]
 \begin{center}
 \includegraphics[scale=.5]{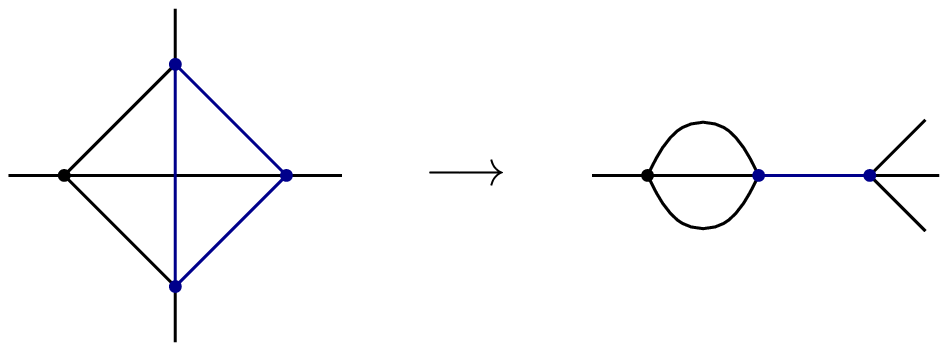}  
 \caption{A $3\to2$ move which generates a melon.} \label{fig:4to1}
 \end{center}
 \end{figure}
The melon can not include the new edge (as this would require some of the 
external half edges of the triangle to be joined together) hence it must include one of the new vertices. We are then 
in the situation depicted in Fig.~\ref{fig:4to1}.
  Eliminating the melon yields a graph $\hat G$ with:
    \[ \omega( \hat G) \le \omega(\tilde G) = \omega(G) -\frac{1}{2} \; .\]
   The new graph $\hat G$ cannot have a tadpole (as in this case $G$ would have a dipole). There are two cases:
    \begin{itemize}
     \item either $\hat G$ has no melon. Then we set $G' = \hat G$. 
     \item or $\hat G$ has a melon. In this case the melon must involve the new vertex and (in $G$) we are in the situation depicted in Fig.~\ref{fig:4to1a}.
 \begin{figure}[htb]
 \begin{center}
 \includegraphics[scale=.6]{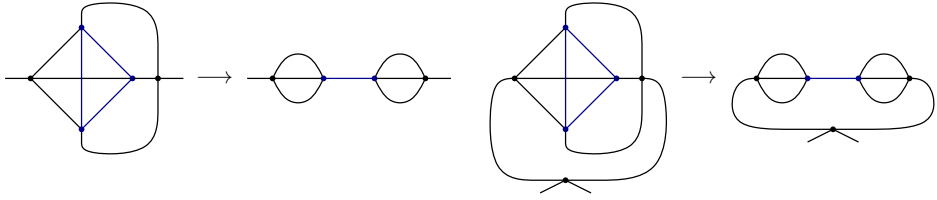}  
 \caption{A $3\to2$ move which generates two melons. In the worst case this require one tadpole deletion.} \label{fig:4to1a}
 \end{center}
 \end{figure}  
 
  We eliminate the new melon to obtain $\hat G_1$ with $\omega(\hat G_1) \le   \omega(\hat G) \le \omega(\tilde G) = \omega(G) - \frac{1}{2}$. The new graph $\hat G_1$ cannot have melons, and it can at most have one tadpole
 (see Fig.~\ref{fig:4to1a}). Eliminating the tadpole if it exists yields a graph $G'$ with:
 \[
  \omega(G') \le \omega(\hat G_1) + \frac{1}{2} \le \omega(G) \;,
 \] 
and $G'$ can not have any more melons or tadpoles.
    \end{itemize}

\item{\it $3\to_R 2$ move.} We set $G' = \tilde G$ as $G'$ can not have melons (otherwise $G$ would have dipoles).

\end{itemize}

\item[There exists a $(1+1+1)$-triangle $H$ with boundary graph $\cG_6$ .] We attempt to delete $H$ by a $3 \to 1$ move shown on the left-hand side of Fig.~\ref{fig:move6}. One can perform this
move in three possible channels, which we label by the edge being created: $(c2)$ (as represented in Fig.~\ref{fig:move6}), $(13)$ or $(ab)$. 

\begin{figure}[htb]
 \begin{center}
 \includegraphics[scale=.6]{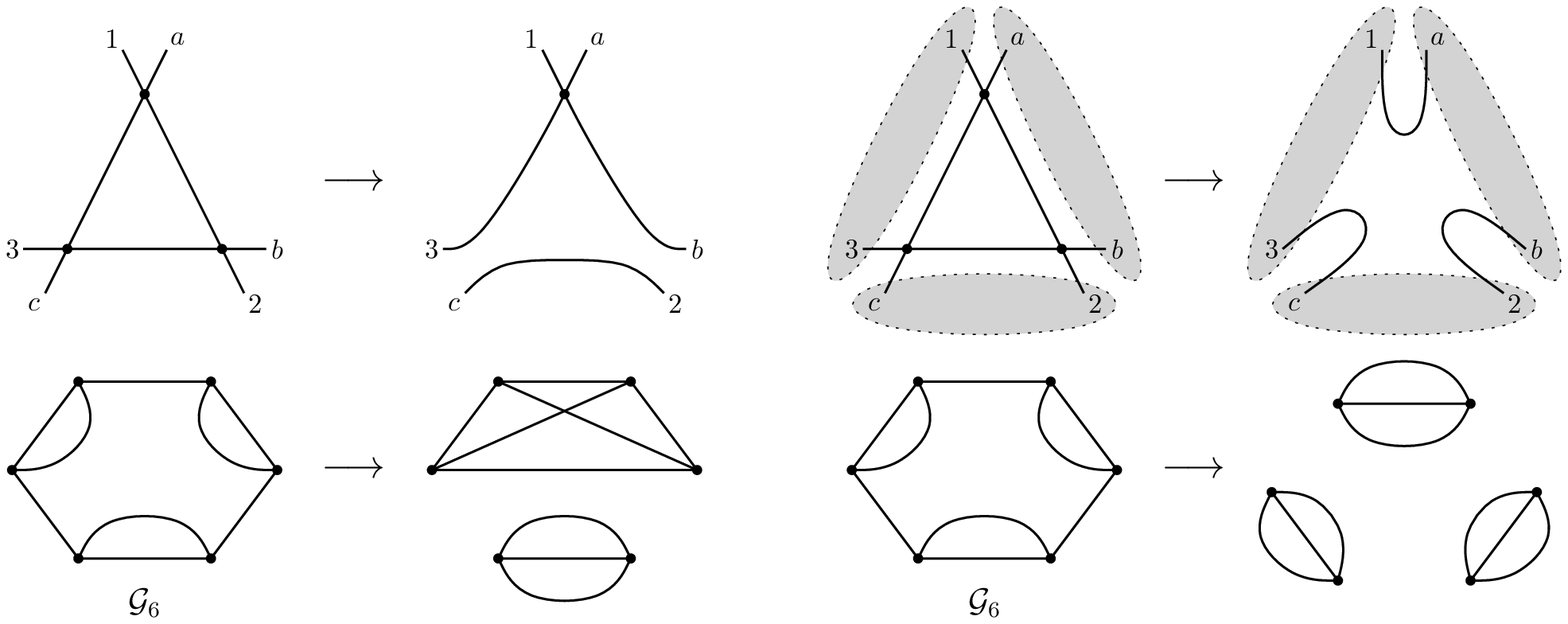}  
 \caption{Deletion of a triangle with boundary graph $\cG_6$: if the $3\to1$ moves of the type shown in the left panel disconnect the graph in all channels, one performs instead the $3\to 0$ move of 
 the right panel, which leaves the diagram connected.} \label{fig:move6}
 \end{center}
 \end{figure}
 
We have two cases:
\begin{itemize}
\item Assume first that performing the $3\to 1$ move in one of the channels (say $(c2)$) yields a connected graph $\tilde{G}$. The latter is obtained after two vertices are deleted and two cut-and-glue operations,
so that $\omega(\tilde{G}) \leq \omega(G) - 2\times 3/2 + 2 = \omega(G) - 1$. Three edges are created in the deletion: $(2c)$ and the edges involving the half edges $3$ and $b$.
\begin{itemize}
\item if $\tilde{G}$ has no tadpole and no melon, we set $G' = \tilde{G}$.
\item if $\tilde{G}$ has a melon, then at least two of the internal edges of the melon must be created by the deletion (otherwise $G$ would have had a dipole). If the new edge $(2c)$ belongs to the melon then, in $G$,
either $2$ or $c$ had to be connected to either $1$ or $a$ (as at least one other internal edge of the melon has been created by the deletion),  and $G$ had a dipole, which is impossible.
Finally, if only the new edges $3$ and $b$ are involved in the melon, then we are in the situation depicted in Fig.~\ref{fig:melcreated}.
\begin{figure}[htb]
 \begin{center}
 \includegraphics[scale=.5]{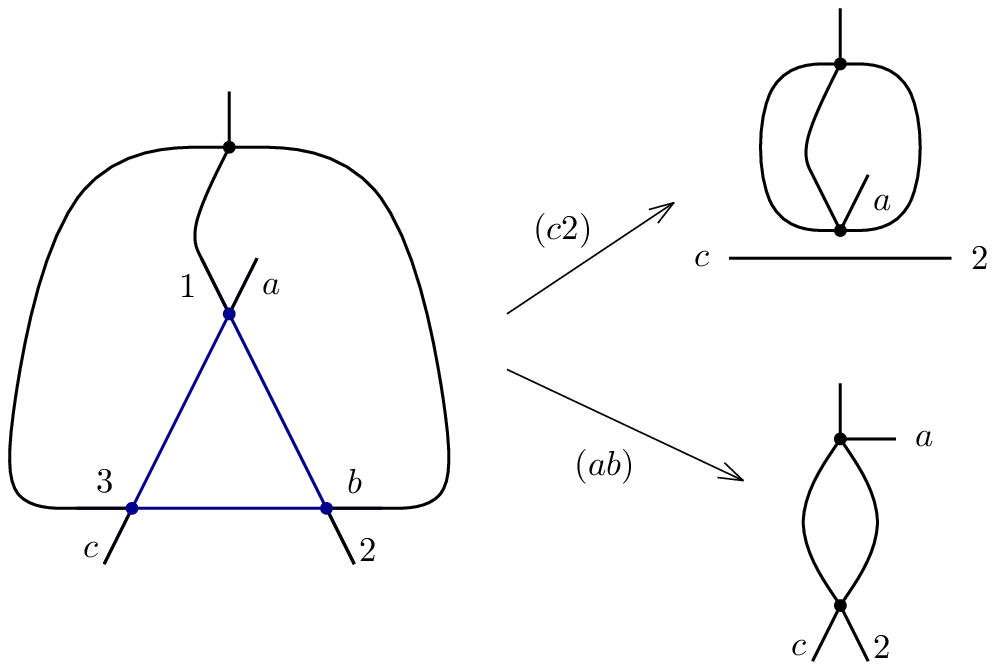}  
 \caption{Possible structure $G$ which generates a melon in the channel $(2c)$. Deleting in the channel $(ab)$ does not create tadpoles or melons.} \label{fig:melcreated}
 \end{center}
 \end{figure}

 In this case we delete in the channel that detaches the external halfedge $a$ of the triangle, $(ab)$. The resulting graph $G'$ can not have tadpoles or melons. 
\item if $\tilde{G}$ has a tadpole, the internal edge of the tadpole can only be the new edge $(c2)$ (since $G$ has no tadpole and no dipole). 
One eliminates it to get $\hat{G}$ with $\omega(\hat{G}) \leq \omega(\tilde{G}) + 1/2 \leq \omega(G) - 1/2$.
This cannot create a new tadpole, as this would require a tadpole or a dipole in $G$. 
\begin{itemize}
\item if $\hat{G}$ has no melon, we define $G' = \hat{G}$.
\item 
if $\hat{G}$ has a melon, the latter must involve at least one of the half edges $3$ or $b$, as otherwise there would be a dipole in $G$. Moreover, in $\tilde{G}$ the tadpole must have been connected to $1$ or $a$ (say $a$),
since a connection to $3$ or $b$ is also forbidden by the absence of dipole in $G$. We are then in one of the two cases shown in Fig.~\ref{fig:pproblem1}, depending on whether only one of the half edges $3$ or $b$ is involved
(in this case, say $b$), or both are. Both of these diagrams are $2$-point graphs, so they can generate at most one extra tadpole upon deletion. Deleting the latter if it occurs, we obtain a graph $G'$ with
$\omega(G') \leq \omega(\hat G ) + 1/2 \leq \omega(G)$.
\end{itemize}
\end{itemize}  

\begin{figure}[htb]
 \begin{center}
 \includegraphics[scale=.4]{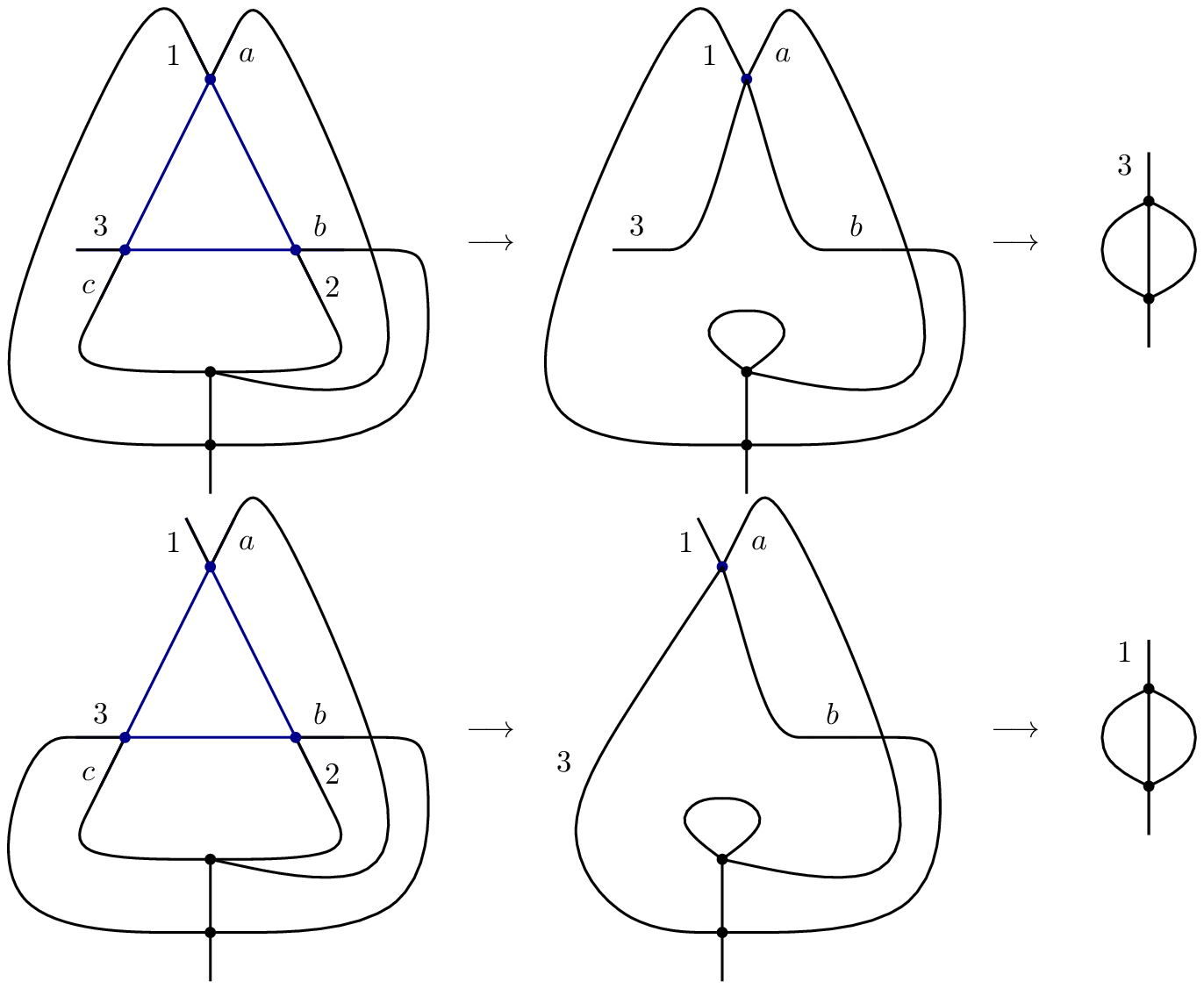}  
 \caption{Possible local structures of $G$, $\tilde{G}$ and $\hat{G}$ when a melon is generated by the deletion of the blue triangle. At most two tadpoles must be eliminated in total.} \label{fig:pproblem1}
 \end{center}
 \end{figure}

\item Assume now that performing a $3\to 1$ move in any of the channels yields a disconnected graph, so that we are in the situation shown on the right panel of Fig.~\ref{fig:move6}. We then instead act on $G$ with the $3\to 0$
move shown on the right-hand side of Fig.~\ref{fig:move6}, to obtain $G'$. $G'$ being formed by a 
gluing of three $2$-point subgraphs of $G$ in a loop, it is connected, and can have neither tadpoles nor melons. Moreover, this $3 \to 0$ move involves four cut-and-glue operations, so
that $F(G') \ge F(G) - 4$ and $\omega(G') \leq \omega(G) - 1/2$. 
\end{itemize}

\item[All triangles are untwisted.] If the triangle is untwisted, then one performs a deletion as depicted in Fig.~\ref{fig:cases3} on the left, to obtain $\tilde{G}$. 
Observe that one can choose the channel of deletion freely (that is the special pair of half-edges hooked to the same vertex which are reglued in an edge, like $(1,a)$ in Fig.~\ref{fig:cases3}). We have several cases:
\begin{itemize}
 \item The graph $\tilde{G}$ is connected. We have $V(\tilde{G}) = V(G)-3$ and $F(\tilde{G}) \ge F(G)- 4$, hence
    \[
     \omega( \tilde G) \le \omega(G) -\frac{1}{2} \;.
    \]
    
   \begin{itemize}  
	\item If $\tilde{G}$ has no melons and no tadpoles, we set $G'=\tilde{G}$.   
   \item {\it $\tilde{G}$ has a melon.} As $G$ has no dipoles, if $\tilde G$ has a melon then at least two of the new edges must be involved in the melon. 
       We have:
       \begin{itemize}
        \item if $(2,c)$ and $(3,b)$ belong to the same melon after deletion, then: either the vertices $(2,b)$ and $(3,c)$ were connected by dipoles to two other vertices in $G$, which cannot be; or we are in the situation 
        depicted in Fig.~\ref{fig:untwisted-samemel}. But then there exists a pair of adjacent triangles in $G$, which are both untwisted and we conclude by Lemma~\ref{lem:4to2}.

\begin{figure}[htb]
 \begin{center}
 \includegraphics[scale=.4]{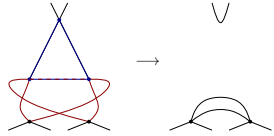}  
 \caption{A deletion of an untwisted triangle (in blue) which generates a melon. The red lines on the left form a pair of adjacent triangles.} \label{fig:untwisted-samemel}
 \end{center}
 \end{figure}  
        \item if $(1,a)$ and $(3,b)$ (resp. $(2,c)$) belong to a melon, but $(2,c)$ (resp. $(3,b)$) does not, then we are in the case depicted in Fig.~\ref{fig:trimelon}. Therefore, there is again a pair of adjacent untwisted
        triangles in $G$ and we can build $G'$ using Lemma~\ref{lem:4to2}. 
        \begin{figure}[htb]
 \begin{center}
 \includegraphics[scale=.4]{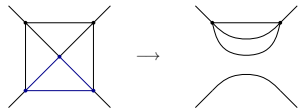}  
 \caption{Triangle deletion which generates a melon.} \label{fig:trimelon}
 \end{center}
 \end{figure}  
  \item if all the new edges are in the melon, we find again pairs of adjacent triangles and we conclude by Lemma~\ref{lem:4to2}
\end{itemize}
   
   \item {\it $\tilde{G}$ has no melon but has a tadpole.} Due to the asymmetry of the deletion, the edges behave differently. The edge $(1,a)$ cannot become a tadpole edge (as $G$ would have a dipole in this case).
     
      The edge $(3,b)$ or $(2,c)$ (or both) can. But then we are in one of the two situations shown in Fig.~\ref{fig:untwisted_tad}, and one finds again a pair of adjacent triangles and we conclude again by Lemma~\ref{lem:4to2}.
 \begin{figure}[htb]
 \begin{center}
 \includegraphics[scale=.3]{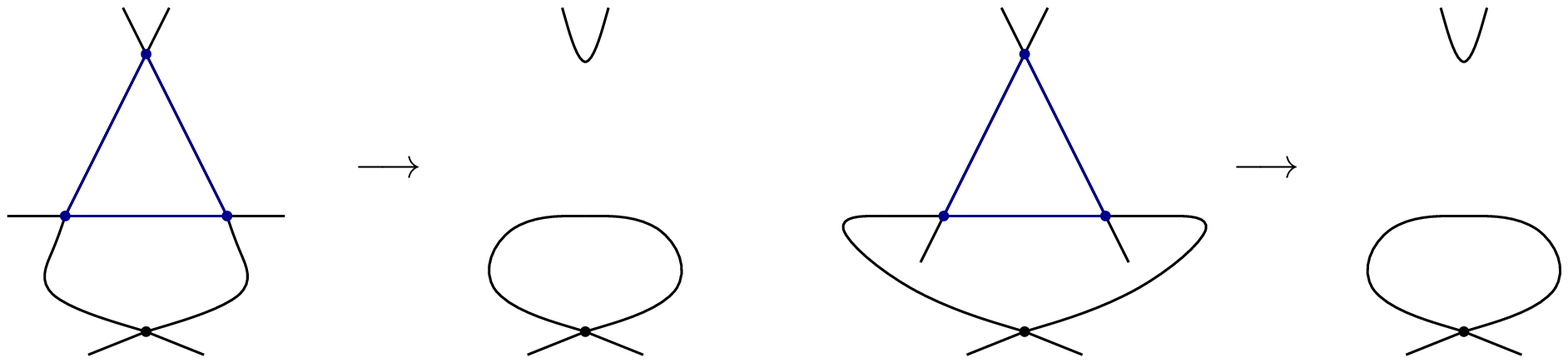}  
 \caption{Triangle deletions which generate one or more tadpoles.} \label{fig:untwisted_tad}
 \end{center}
 \end{figure}  
   \end{itemize}

\item The graph $\tilde{G}$ is disconnected. It has then two or three connected components, we discuss each situation separately. 
   \begin{itemize}
    \item If $\tilde{G}$ has three connected components, then $(1,a)$, $(3,b)$ and $(2,c)$ are all two point graphs. But then, deleting the triangle in the channel $(b,2)$, $(c,a)$ and $(1,3)$ yields a 
    connected graph $\tilde{G}_1$ from which we construct a suitable $G$ as previously described.
     \item If $\tilde{G}$ has two connected components, one of the alternative channels may yield a connected diagram $\tilde{G}_1$, which we can use to construct $G'$ as previously explained. 
     If none of the channels lead to a connected graph, $G$ is necessarily of the form presented in Fig.~\ref{fig:untwisted_final}.
\begin{figure}[htb]
 \begin{center}
 \includegraphics[scale=.6]{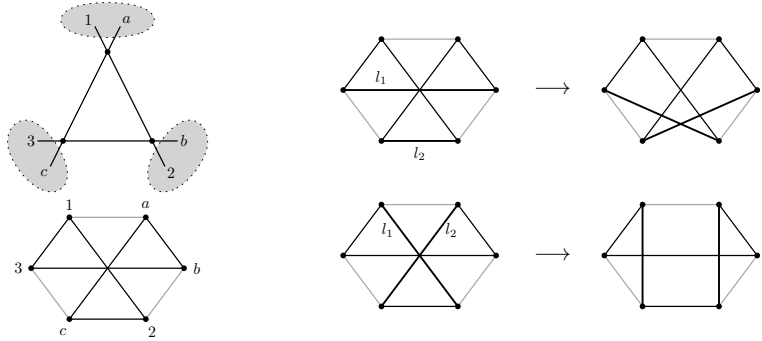}  
 \caption{Untwisted triangle which always leads to a disconnected graph upon deletion. One then performs a cut-and-glue operation to obtain a prism boundary graph.} \label{fig:untwisted_final}
 \end{center}
 \end{figure}  
On the left panel, we have represented the triangle and its boundary graph. 
There exists a pair of black edges $(l_1 , l_2)$ in the boundary graph which are:
\begin{itemize}
 \item traversed by the same face of $G$,
 \item completely disjoint (that is they do not share any vertex)
\end{itemize}

 To see this, consider the face $f$ associated to the boundary edge $(3b)$. As it goes into the connected component attached to the corner $(3,c)$, it must also come out of it:
\begin{itemize}
 \item if it exits through the half edge $c$ we are done as it traverses as either $(2c)$ or $(ac)$.
 \item if it exists through $(31)$, it enters the connected component attached to $(1,a)$ and exists through
   \begin{itemize}
    \item either $(12)$ which is disjoint from $(3b)$,
    \item or $(ac)$ or $(ab)$ which are disjoint from $(13)$.
   \end{itemize}
\end{itemize}

One performs a cut-and-glue operation on the pair $(l_1, l_2)$ which turns the boundary graph into a prism. As this either increases the number of faces or leaves it unchanged, the degree 
does not increase. One then performs a $3\to2$ move to obtain the graph $G'$ with degree $\omega(G') \le \omega(G) - 1/2$ (as the internal face of the triangle and one vertex are deleted).
This is illustrated on the right panel of Fig.~\ref{fig:untwisted_final}.

  \end{itemize}

\end{itemize}

\end{description} 

\end{proof}

\newpage

\section{Leading order}\label{sec:LO}

Let us first observe that since  Eq.~\ref{eq:crrrucial} contains no tadpoles and no melons, Theorem \ref{thm:LO} is equivalent to the following statement: the only connected stranded graph that survives in Eq.~\ref{eq:crrrucial} in the large $N$ limit is the ring graph of degree zero. All the tadpoles and melons have in fact been resummed into the new  covariance $K (\lambda, N)$, and the fact that in the large-$N$ limit the latter reduces to the generating function of the $4$-Catalan numbers means that only melons contribute to it at leading order in $1/N$ \cite{RTM}.
 
A natural conjecture would be that all stranded diagrams with no melon and no tadpole have strictly positive degree. However, this turns out to be wrong: for instance, one may easily construct stranded configurations with vanishing degrees starting from the graph $H_2$ closed onto itself (see Fig. \ref{fig:adegchains}, as well as Fig. \ref{fig:double-triangle} below). 

We therefore adopt a more refined strategy which will consist in: 1) identifying a suitable family of stranded diagrams for which $\omega \geq 1/2$ always holds; 2) proving that non-trivial cancellations occur for all the remaining graph configurations (which may have a vanishing degree).  

\

We first deal with the trivial situation in which there are no short faces.

\begin{lemma}\label{lem:LO-noshort}
Let $G$ be a stranded graph with no melon, no tadpole, no dipole and no triangle. If $\omega(G) =0$, then $G$ is the ring graph of vanishing degree.
\end{lemma}
\begin{proof}
Suppose $G$ is not a ring graph. We can replace all the broken edges of $G$ by unbroken ones to obtain $\tilde{G}$ such that $\omega(G) \geq \omega(\tilde{G})$. Furthermore $\tilde{G}$ cannot have short faces, hence its degree is strictly positive. 
\end{proof}

\

We then make use of the irreducible character of the antisymmetric and symmetric traceless $O(N)$ representations to fix the structure of an arbitrary connected $2$-point function. 
\begin{lemma}\label{lemma:2-point-N}
Let $\cG$ be a connected (and non amputated) $2$-point graph. The associated amplitude $A(\cG)_{a_1 a_2 a_3 , b_1 b_2 b_3}$ can be written as:
\be
A(\cG)_{a_1 a_2 a_3 , b_1 b_2 b_3} = \lambda^{V(\cG)} f_\cG (N) \, {\bf P}_{a_1 a_2 a_3 , b_1 b_2 b_3}\,, \nonumber
\ee
where $f_\cG (N)$ is uniformly bounded.
\end{lemma}
\begin{proof}
In what follows, we use the short-hand $A_{a_1 a_2 a_3 , b_1 b_2 b_3}$ for $A(\cG)_{a_1 a_2 a_3 , b_1 b_2 b_3}$. It defines a bilinear form $A(\cdot , \cdot)$ on the space of tensors: 
$$A(T,T') = T_{a_1 a_2 a_3} A_{a_1 a_2 a_3 , b_1 b_2 b_3} T'_{b_1 b_2 b_3}$$ 
By construction, $A(T,T')$ is a sum over stranded configurations, all of which contract the indices of $T$ and $T'$ pairwise (along external strands), in a $O(N)$ invariant way. It follows that $A(T,T')$ is invariant under $O(N)$: 
$$\forall O \in O(N), \qquad A(O\cdot T  , O \cdot T') = A(T, T')$$
By duality with respect to the standard inner product on the space of tensors\footnote{In the notation of this paper, the inner product of $T$ and $T'$ is defined as $\langle T \vert T' \rangle := T {\bf 1} T'$}, one can construct a map $\hat{A}$ from the space of rank-$3$ tensors to itself. In an arbitrary orthonormal basis $\{ T_n \}$, it takes the form:
$$
\hat{A} (T) := \sum_{n} A(T_n ,T) T_n  
$$
We now show that: a) $\hat{A}(\mathrm{Ker} \, {\bf P}) = \{ 0 \}$ ; b) $\hat{A}$ defines a $O(N)$ intertwiner; c) $\mathrm{Im} \, \hat{A} \subset \mathrm{Im} \, {\bf P}$. 

First, since the graph $\cG$ is not amputated (and ${\bf P}$ is symmetric) one can decompose $A(T,T')$ as: $A(T,T') = \tilde{A} ( {\bf P} T , {\bf P}  T' )$, where $\tilde{A}$ is a bilinear form. Hence $A(T,T')=A({\bf P} T,T') = A(T,{\bf P} T')$. In particular, for any $T \in \mathrm{Ker} \, {\bf P}$:
$$
\hat{A} (T) = \sum_n A(T_n ,  T) T_n = \sum_n A(T_n , {\bf P} T) T_n = \sum_n A(T_n , 0) T_n = 0\,.  
$$

Second, the covariance of $\hat{A}$ is a direct consequence of the invariance of $A$. Indeed, for any orthogonal transformation $O \in O(N)$ we show that
$$ 
\hat{A} (O \cdot T) = \sum_n A(T_n , O \cdot T) T_n = \sum_n A( O^{-1} \cdot T_n , T) T_n = \sum_n A(T_n^O , T) O \cdot T_n^O = O \cdot \hat{A} (T)\,,
$$
where we have introduced $T_n^O := O^{-1} \cdot T_n$. The invariance of $A$ has been invoked in the second equality. In the last equality, we have used the fact that the $O(N)$ action preserves the inner product to conclude that $\{ T_n^O \}$ is an orthonormal basis. 

Finally, to prove that $\mathrm{Im} \, \hat{A} \subset \mathrm{Im} \, {\bf P}$ it is convenient to choose an orthornomal basis adapted to the orthogonal projector $\mathbf{P}$. Calling $p$ the dimension of the full tensor space and $q$ the dimension of $\mathrm{Im} \, {\bf P}$, we assume that $\{ T_1 ,\ldots , T_q \}$ is an orthonormal basis of $\mathrm{Im} \, {\bf P}$, and that $\{ T_{q+1} ,\ldots , T_p \}$ is an orthonormal basis of $\mathrm{Ker} \, {\bf P}$. For any tensor $T$, we then have:
$$
\hat{A}(T) = \sum_{n=1}^p A(T_n , T) T_n = \sum_{n=1}^p A({\bf P} T_n , T) T_n = \sum_{n=1}^q A( T_n , T) T_n \in \mathrm{Im} \, {\bf P}
$$  

From b) and c), we conclude that $\hat{A}$ induces an intertwining map from the image of ${\bf P}$ to itself. The antisymmetric and symmetric traceless representations being irreducible (for large enough $N$) \cite{weyl1946, hamermesh}, Schur's Lemma implies that $\hat{A}$ is proportional to the identity on $\mathrm{Im} \, {\bf P}$. Since it furthermore vanishes on $\mathrm{Ker} \, {\bf P}$, one concludes that $\hat{A}$ is a multiple of ${\bf P}$. 

Hence $A_{a_1 a_2 a_3 , b_1 b_2 b_3}$ is proportional to ${\bf P}_{a_1 a_2 a_3, b_1 b_2 b_3}$. Finally, the coefficient of proportionality between these two quantities must be uniformly bounded in $N$, otherwise the $1/N$ expansion would not exist, contradicting Theorem~\ref{thm:main}.  
\end{proof}

\

In view of Lemma \ref{lemma:2-point-N}, the power-counting arguments of the previous sections can be immediately generalized to graph configurations in which the bare propagators are substituted with arbitrary connected $2$-point subgraphs. 
It will in particular be convenient to consider melons and tadpoles with such decorations, as represented in Fig. \ref{fig:gen-tadmel}. We call them \emph{generalized tadpoles} and \emph{generalized melons}.  
\begin{figure}[htb]
 \begin{center}
 \includegraphics[scale=.7]{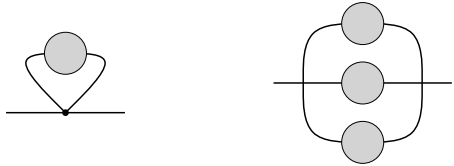}  
 \caption{A generalized tadpole (left) and a generalized melon (right).} \label{fig:gen-tadmel}
 \end{center}
 \end{figure}

\

In the following, we will set-up an induction on the number of vertices of a graph with no generalized melon and no generalized tadpole\footnote{A graph with no generalized melon and no generalized tadpole has in particular no melon and no tadpole, so the degree of any of its configuration is positive or zero.}. The next Lemma deals with the smallest graph in this family. 
\begin{lemma}\label{lem:init}
The smallest vacuum diagram with no generalized melon, no generalized tadpole and $V(\cG) \geq 1$ is the double triangle graph of Fig. \ref{fig:double-triangle}. Any of its stranded configurations has degree $\omega \geq 1/2$.
\end{lemma}
\begin{proof}
A vacuum diagram with one or two vertices is either a melon or a tadpole, and the double triangle graph is the only suitable configuration on three vertices. To determine its scaling in $N$, we remark that it is nothing but $H_2$ closed onto itself. By Lemma \ref{lem:2pointmoves}, the degree of any of its stranded configurations is at least $1/2$. 
\end{proof}

\begin{figure}[htb]
 \begin{center}
 \includegraphics[scale=.7]{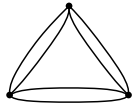}  
 \caption{The double triangle graph, obtained by closing $H_2$ onto itself.} \label{fig:double-triangle}
 \end{center}
 \end{figure}

\

In the following, we will implicitly rely on the following simple observation, which we state without proof.
\begin{lemma}\label{lem:gen_comb}
Let $\cG$ be a vacuum graph and $\cH$ a $2$-point graph, both without generalized melon and without generalized tadpole. Replacing any line of $\cG$ by $\cH$ yields a graph $\tilde{\cG}$ which has itself no generalized melon and no generalized tadpole.
\end{lemma}
\noindent From a combinatorial perspective, this will allow to analyze the occurrence of generalized melons and tadpoles in a graphical way, in the exact same manner as for ordinary melons and tadpoles.  

\

We are now in the position to prove that a stranded graph with vertices but no generalized tadpole and no generalized melon cannot contribute to the leading order.
\begin{proposition}\label{prop:LO-bound}
Let $G$ be a stranded graph with $V(G) \geq 1$. If $G$ has no generalized tadpole and no generalized melon, then $\omega(G) \geq 1/2$.
\end{proposition} 
\begin{proof}
We proceed by induction on the number of vertices. The initialization is provided by Lemma \ref{lem:init}. Furthermore, Lemma \ref{lem:LO-noshort} ensures that the proposition holds for any graph with no dipole or triangle.

If $G$ has more than $3$ vertices and contains short faces, we look for deletions of dipoles or triangles, following the exact same steps as in the previous sections. The only differences is that some propagator lines may now be decorated with arbitrary connected $2$-point insertions, and that the words "tadpole" and "melon" must be replaced by "generalized tadpole" and "generalized melon". In particular, we must be able to delete any generalized tadpole or generalized melon that may be generated by the deletion of a dipole or a triangle. From a combinatorial point of view, this does not add any difficulty, thanks to the simple observation of Lemma \ref{lem:gen_comb}. Hence, the result of the previous sections immediately allow to construct a stranded configuration $G'$ with: strictly fewer vertices than $G$, no generalized tadpole or generalized melon, and degree verifying $0 \leq \omega(G') \leq \omega(G)$. To conclude, we claim that we are in one of three situations:
\begin{enumerate}
\item the derived bound is actually strict, i.e. $\omega(G')< \omega(G)$, in which case we immediately obtain $\omega(G) \geq 1/2$;
\item $V(G')>0$, in which case we can apply the induction hypothesis to deduce that $\omega(G) \geq \omega(G') \geq 1/2$;
\item $\cG$ has a generalized tadpole, which contradicts our working assumption. 
\end{enumerate}
Checking the validity of this claim is straightforward but tedious, as it requires to go once more through the combinatorial analysis of Section \ref{sec:proof1}. We limit ourselves to listing all the configurations for which we have to resort to 2. or 3.
\begin{itemize}

\item Deletion of dipoles

\begin{itemize}
\item There exists a deletion in the $\perp$ channel which disconnects $\cG$. Then $\cG$ has a a generalized tadpole.

\item There exists a deletion in the $\perp$ channel which does not disconnect the graph, and does not create generalized tadpoles or generalized melons. Then $G'$ is obtained by performing the deletion. Furthermore $V(G')>0$, otherwise $G$ would be a melon on two vertices. 

\item All $\perp$ deletions do not disconnect the graph, and create at least a melon or a tadpole. Then $G$ must contain one of the patterns shown in Fig. \ref{fig:adelchains} (up to some connected $2$-point decorations). In particular, $V(G) \geq 3$. 

\begin{itemize}
\item If there exists a deletion in the $=$ (or $\times$) channel which does not create a generalized tadpoles or melons, one can perform it and obtain a graph $G'$ with $V(G')>0$. 
\item If not, we are in one of the situations shown in Fig. \ref{fig:adegchains} (up to some connected $2$-point decorations). Performing a (generalized) $H$-contraction followed by generalized tadpole and melon contractions when necessary, we obtain a suitable graph $G'$. The only situations in which this graph $G'$ may not have a strictly smaller degree compared to $G$ is in the presence of a (generalized) $H_2$ $2$-point subgraph closed into a generalized tadpole. But then we are in situation 3.
\end{itemize}
\end{itemize}

\item Deletion of triangles
\begin{itemize}
\item Deletion of a pair of adjacent untwisted triangle (Lemma \ref{lem:4to2}).
\begin{itemize}
\item The successive deletions generate a generalized tadpole, a generalized melon, and another generalized tadpole. We are then in the situation of Fig. \ref{fig:2twisted_tad}, hence $G$ has a generalized tadpole. 
\item The successive deletions generate a generalized melon, a generalized tadpole, and another generalized tadpole. Then $G$ has a generalized tadpole, made out of one of the two $2$-point structures shown on the right panel of Fig. \ref{fig:2twisted_mel}. 
\end{itemize}
\item Deletion of triangles in graphs with no generalized melon, no generalized tadpole, and no dipole (Lemma \ref{lem:tripole}).
\begin{itemize}
\item Deletion of a triangle which is neither untwisted nor with boundary graph $\cG_6$. One may successively generate two generalized melons and one generalized tadpole, in which case the bound is not strict. But then there is a generalized tadpole in $G$, of the type shown on the right side of Fig. \ref{fig:4to1a}.
\item Deletion of a $1+1+1$-triangle with boundary $\cG_6$. The obtained bound on the degree is always strict, unless we have a generalized tadpole with one of the two $2$-point structures depicted in Fig. \ref{fig:pproblem1}.
\end{itemize}
\end{itemize}
This concludes the proof. 
\end{itemize}

\end{proof}

\

We are left to discuss graphs which contain generalized tadpoles. The subtlety is that this family features stranded configurations with degree $0$, but which nonetheless do not contribute to the leading order. This is due to similar cancellations as already identified in graphs with tadpoles. 
\begin{lemma}\label{lem:LO-tad}
Let $\cG$ be a connected and vacuum Feynman graph. If $\cG$ has a generalized tadpole, then it is not leading-order. 
\end{lemma}
\begin{proof}
In view of Lemma \ref{lemma:2-point-N}, the same cancellations found in ordinary tadpoles also occur in generalized tadpoles. Hence, if $\cG$ has a generalized tadpole, its amplitude must be suppressed by a factor $N^{-1/2}$ at least. 

In more detail, $\cG$ must have the structure shown in Fig. \ref{fig:lem-gen-tad}, where $\mathcal{A}$ and $\cB$ are connected $2$-point graphs. By Lemma \ref{lemma:2-point-N}, there exists two bounded functions $f_\mathcal{A}$ and $f_{\cB}$ such that:  
$$
A(\cG) = \lambda^{V(\mathcal{A})+ V(\mathcal{B})} f_\mathcal{A}(N) f_{\cB}(N)  \bP_{a_1a_2a_3, c_1c_2c_3} \bP_{c_3c_4c_5, c_5c_2c_6} \bP_{c_6 c_4 c_1,a_1 a_2 a_3} 
$$
The computations of Section \ref{sec:subtr} imply that this quantity scales as:
$$
f_\mathcal{A}(N) f_{\cB}(N) f_1^\bP (N) \bP_{a_1a_2a_3, a_1a_2a_3} = O(N^{-1/2}) \times \bP_{a_1a_2a_3, a_1a_2a_3}
$$
and therefore decays as $N^{-1/2}$ or faster.
\begin{figure}[htb]
 \begin{center}
 \includegraphics[scale=.7]{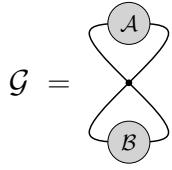}  
 \caption{Structure of the connected graph $\cG$ of Lemma \ref{lem:LO-tad}.} \label{fig:lem-gen-tad}
 \end{center}
 \end{figure}
\end{proof}

\

Lastly, we can now prove the following proposition.

\begin{proposition}\label{propo:LO-ring}
Let $\cG$ be a connected graph with no melon and no tadpole. If $\cG$ is leading order, then it is the ring graph. 
\end{proposition}
\begin{proof}
By Proposition \ref{prop:LO-bound} and Lemma \ref{lem:LO-tad}, if $\cG$ is leading-order and have vertices, then it must have a generalized melon. $\cG$ having no melon, at least one of the three $2$-point functions making the generalized melon must contain vertices. It is also easy to see that this $2$-point function must itself be leading-order, and hence contain a generalized melon. By induction, we are thus led to the absurd conclusion that $\cG$ contains infinitely many generalized melons. 
\end{proof}
For any graph $\cG$ with no tadpole and no melon, we have thus shown that:
\begin{itemize}
\item either $\cG$ is a ring graph, in which case it contributes to the leading order through its unique stranded configuration of degree $0$;
\item or $\cG$ has vertices but no generalized tadpole, in which case any of its stranded configurations $G$ has degree $\omega(G) \geq 1/2$;
\item or $\cG$ has a generalized tadpole, in which case it may have stranded structures with vanishing degree, but they compensate and $A(\cG)$ still decays like $N^{-1/2}$ or faster.  
\end{itemize} 

This achieves the proof of Theorem \ref{thm:LO}

\cleardoublepage

\appendix

\section*{Acknowledgements}

\noindent This research was supported in part by Perimeter Institute for Theoretical Physics. Research at Perimeter Institute is supported by the Government of Canada 
through the Department of Innovation, Science and Economic Development Canada and by the Province of Ontario through the Ministry of Research, Innovation and Science.

\section{Gaussian integral}\label{app:gauss}

We consider the Gaussian expectation of the invariant observable:
\begin{align}\label{eq:unu}
 \left[ e^{\frac{1}{2} \;  \partial_T \bP \partial_T    }      \;  T {\bf 1} T \right]_{T=0} = \sum_{\genfrac{}{}{0pt}{}{a_1,a_2,a_3}{b_1,b_2,b_3}}
  \bP_{a_1a_2a_3,b_1b_2b_3} {\bf 1}_{a_1a_2a_3,b_1b_2b_3} \;. 
\end{align}

Eq.~\eqref{eq:unu} has a convenient graphical representation. Each of the six (respectively fifteen) terms can be represented as a \emph{ring graph} consisting in an edge with three strands closing onto itself, as
depicted in Fig.~\ref{fig:ring}.
\begin{figure}[htb]
 \begin{center}
 \includegraphics[width=12cm]{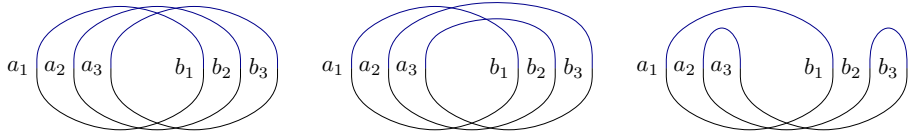}  
 \caption{Ring graphs representing the terms $\delta_{a_1b_1}\delta_{a_2b_2} \delta_{a_3b_3} $, $\delta_{a_1b_1}\delta_{a_2b_3} \delta_{a_3b_2} $ and $\delta_{a_1b_1}\delta_{a_2a_3} \delta_{b_2b_3} $ .} \label{fig:ring}
 \end{center}
 \end{figure}
The strands represent the indices of the tensors. The lower part of the drawing has three strands going from one end to the other and represents the $\delta_{a_i b_i} $ identifications.
As we represent $ T_{a_1a_2a_3} {\bf 1}_{a_1a_2a_3,b_1b_2b_3} T_{b_1b_2b_3}$, the three lower strands cross ($a_1$ goes to $b_1$, $a_2$ to $b_2$ and $a_3$ to $b_3$). 

The upper part of the drawing represents one of the terms in $\bP_{a_1a_2a_3,b_1b_2b_3}  $. If $\bP = \bA$ we have six terms, each of them 
identifying the indices $b$ with a permutation of the indices $a$. They are represented as three strands going from $a$ to $b$ via a permutation (the first two cases in  Fig.~\ref{fig:ring}).
In the symmetric traceless case $\bP = \bS$ and we get nine extra terms. Each of the new terms identifies and index $a$ with an index $b$, and then identifies pairwise the two remaining indices $a$ and the two 
remaining indices $b$ (the last case in  Fig.~\ref{fig:ring}).
Each ring graph comes with a global sign, and each closed strand, which we call a \emph{face}, brings a free sum, hence a factor $N$. 
\begin{description}
 \item[\it The antisymmetric case.] We have six terms corresponding to six permutations:
     \begin{itemize}
      \item[--] the identity permutation has a $+$ sign and $3$ faces, hence contributes $ \frac{1}{3!} N^3$,
      \item[--] three odd permutations with a $-$ sign and $2$ faces, bringing $\frac{1}{3!} (-3N^2) $,
      \item[--] another two even permutations with a $+$ sign and $1$ face, bringing $\frac{1}{3!} (2N)$,
     \end{itemize}
   hence a total of:
    \[
      \left[ e^{\frac{1}{2} \;  \partial_T \bA \partial_T    }      \;  T {\bf 1} T \right]_{T=0} =\frac{1}{3!} N(N-1)(N-2) \;.
    \]
 \item[\it The symmetric traceless case.] The six terms corresponding to permutations of the strands add up to:
     \[ 
      \frac{1}{3!} N(N^2 + 3N  + 2) \;,
     \]
     and the nine extra terms bring:
      \[
       \frac{1}{3!} \left( -\frac{2}{N+2} \right) 3N (N+2) \;,
      \]
     therefore we obtain a total of:
      \[
       \left[ e^{\frac{1}{2} \;  \partial_T \bS \partial_T    }      \;  T {\bf 1} T \right]_{T=0} = \frac{1}{3!}  N(N^2 + 3N  - 4) \;.
      \]

\end{description}
 
    Observe that, as expected, in both cases the expectation of $T {\bf 1} T$ is just the number of independent components of the tensor: $\binom{N}{3}$ in the antisymmetric case and:
    \[
     \underbrace{ \binom{N}{3} +  2 \binom{N}{2} + N }_{\text{symmetric tensors} } -\underbrace{ N }_{\text{traceless conditions}} 
    \]
    in the symmetric traceless one.

\section{The symmetric model and the trace instability}\label{app:symmodel}

The $1/N$ expansion does not work in the case of a symmetric tensor with no tracelessness condition because of an instability in the trace modes.
  Let us denote $\tilde \bS$ the projector on symmetric tensors and $\bQ$ the projector on the trace modes:
 \begin{align*}
& \tilde \bS_{a_1a_2a_3, b_1b_2b_3}  =  \frac{1}{3!} \bigg[ 
  \delta_{a_1b_1} ( \delta_{a_2b_2} \delta_{a_3 b_3} + \delta_{a_2b_3} \delta_{a_3b_2}  ) + 
         \delta_{a_1b_2} (  \delta_{a_2b_1} \delta_{a_3 b_3} +\delta_{a_2b_3} \delta_{a_3b_1}  )  +
      \delta_{a_1b_3} (\delta_{a_2 b_1} \delta_{a_3 b_2} +  \delta_{a_2 b_2} \delta_{a_3 b_1} )   \bigg] \crcr
& \bQ_{a_1a_2a_3, b_1b_2b_3} = \frac{1}{3( N + 2) } \bigg( 
   \delta_{a_1b_1} \delta_{a_2a_3} \delta_{b_2b_3} + \delta_{a_1b_2} \delta_{a_2a_3} \delta_{b_1b_3} + \delta_{a_1b_3} \delta_{a_2a_3} \delta_{b_1b_2} + (a_1 \leftrightarrow  a_2) + (a_1 \leftrightarrow  a_3)
 \bigg)\; .
\end{align*}
We have $\tilde \bS \bQ= \bQ \tilde \bS = \bQ$ and the projector on the symmetric traceless part in Eq.~\eqref{eq:S} is $ \bS = \tilde \bS - \bQ $. 
Computing the tadpole and melon corrections for the model with propagator $\tilde \bS$ yields:
\begin{align*}
&  \sum_{c} \tilde\bS_{a_1a_2a_3, c_1c_2c_3} \tilde \bS_{c_3c_4c_5, c_5c_2c_6} \tilde \bS_{c_6 c_4 c_1,b_1b_2b_3} \crcr
 & \qquad \qquad =  \frac{N+2}{ 6 }  \; \tilde \bS_{a_1a_2a_3, b_1b_2b_3} + \frac{(N+2)^2}{18}  \bQ_{a_1a_2a_3, b_1b_2b_3} \crcr
& \sum_{c,d} \tilde \bS_{a_1a_2a_3, c_1 c_2c_3}   \tilde \bS_{c_3 c_4 c_5, d_3 d_4 d_5 } \tilde \bS_{c_5c_2 c_6 , d_5 d_2 d_6 } \tilde \bS_{c_6c_4c_1 ,  d_6 d_4 d_1  }\tilde \bS_{ d_1 d_2 d_3 , b_1b_2b_3}  \crcr
 & \qquad \qquad = \frac{ N^3 + 9N^2 +34N +64 }{6^3} \; \tilde \bS_{a_1a_2a_3, b_1b_2b_3} +  \frac{ ( N^2+9N+26 ) (N+2)}{6^4}    \bQ_{a_1a_2a_3, b_1b_2b_3} \;,
\end{align*} 
and the self energy at second order is 
$\Sigma^{(2)}  =  \left(   \lambda  N^{ - 1/2}  +    \lambda^2  \right) \tilde \bS +   \left(    N^{1/2} \lambda  +  \lambda^2  \right)  \bQ  \sim  \lambda N^{1/2} \bQ$.
The effective two point function (i.e. the would be renormalized covariance) is formally:
\[  \left[ \sum_{q\ge 1} (\lambda N^{1/2})^q  \right] {\bf Q}\;, \]
which is not summable for large $N$. If one keeps $ \lambda N^{1/2}\le 1$ in the large $N$ limit, the series becomes summable, but this suppresses the melonic graphs.

\section{Special cases of Lemma~\ref{lem:2pointmoves} }\label{app:Specialcases}

We discuss the four graphs separately. 

 \begin{figure}[htb]
 \begin{center}
 \includegraphics[scale=.6]{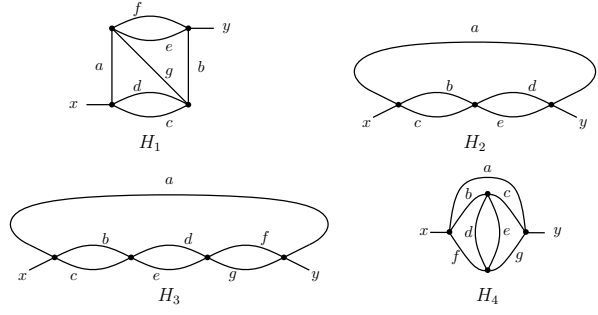}  
 \caption{The four special cases.} \label{fig:adegchainss}
 \end{center}
 \end{figure}

The external leg labels $x$ and $y$, and the edge labels $a,\, b, \, c\,, \ldots, \,g$ refer to Fig.~\ref{fig:adegchainss}. We call $V$, $L$, $F$ (resp. $F_i$) the number of vertices, edges, closed faces (resp. closed faces of length $i$)
in the subgraph. We also define $\ell$ as the sum of the lengths of the open strands of the subgraph. The sum of the lengths of the internal faces of the subgraph is then\footnote{For a $4$-valent $2$-point graph, $4V = 2 L + 2$.}:
\be
S := 3 L - \ell = 6 V - 3 - \ell\,.
\ee

From $F = \sum_i F_i$ and $S = \sum_i i F_i$, one obtains the general combinatorial bounds
\be\label{ineq_H}
\forall k \geq 2\;, \qquad F \leq \lfloor \frac{S}{k+1} + \sum_{i\leq k} \frac{k+1-i}{k+1} F_i \rfloor \;,
\ee
which we will repeatedly use. Since the graphs we will consider have no tadpoles, we will always have $F_1=0$. For $k=2$ we obtain $F \leq \lfloor (S+F_2) /3 \rfloor$, while for $k=3$ we have $F \leq \lfloor ( S + 2 F_2 + F_3) /4 \rfloor $.

\

\noindent {\bf Graph $H_1$.} $V= 4$ and therefore $S = 21 - \ell$.
\begin{itemize}

\item Suppose first that all the external strands traverse from $x$ to $y$ (i.e. the effective $2$-point function is unbroken). The three external strands have length $2$ or more. Moreover, if an external strand follows the path $(x a e y)$ (resp. $(x d b y)$), no external strand can follow the path $(x a f y)$ (resp. $(x c b y)$) (since there is only one corner $(x a)$, and one corner $(b y)$). Hence there can be at most two external strands of length $2$, which implies $\ell \geq 2+2+3= 7$ and $S \leq 14$. There are only two cycles of length $2$ in $H_1$, so that $F_2 \leq 2$. The inequality \eqref{ineq_H} (for $k=2$) yields $F \leq  \lfloor (14 + 2)/3 \rfloor = 5$. Contracting $H_1$ deletes $4$ vertices and therefore 
$\omega \to \omega' \geq \omega - 4\times 3 / 2 + 5 =  \omega - 1$.

\item Suppose now that two external strands loop back. In this case we need to make sure that $F \leq 4$, to compensate for the additional face which may be deleted when replacing the effective broken edge by an unbroken one. Notice that the geometry of the diagram imposes $\ell \geq 7$ and $F_2 \leq 2$.\begin{itemize}
\item If $(\ell = 7, F_2 \geq 1)$, we can assume without loss of generality that $(cd)$ forms a face of length $2$. Up to a permutation of $c \leftrightarrow d$, the only way to maintain $l=7$ is with the following combination of external faces: $(xcgax)$, $(yfey)$ and $(xdby)$. But then it is easy to see that $F_2=1$, and that there must be a face of length $4$ or more\footnote{For instance, a face running through $f$ must necessarily go through $a$, then through $c$ or $d$, then through $b$ or $g$.}. Hence $S - 3 F = -F_2 + F4 + 2 F_5 + 3 F_6 + \ldots \geq 0$, from which we finally obtain $F \leq \lfloor 14/3 \rfloor = 4$.

\item If $(\ell = 8, F_2 = 2)$, one can assume without loss of generality that the external faces are: $(xcgax)$, $(yfgby)$ and $(xdby)$. One finds again that there must be be a face of length $4$ or higher; this yields $S-3F \geq - F_2 + F_4 + \ldots \geq -1 $ and therefore $F \leq \lfloor (13 +1)/3 \rfloor = 4$.

\item In all other cases, $F_2 \leq \ell- 7$ and we immediately obtain $F \leq  \lfloor ( 21 - \ell + \ell - 7)/3 \rfloor = 4$ from \eqref{ineq_H}. 
\end{itemize}

\end{itemize}

\noindent {\bf Graph $H_2$.}
There are $3$ vertices, which gives $S = 15 - \ell$. 

\begin{itemize}

\item When all the external strands traverse from $x$ to $y$, $\ell \geq 2 +2+1= 5$. There are only two cycles in $H_2$, so that $F_2 \leq 2$. Hence $F \leq \lfloor (S + F_2 )/3 \rfloor \leq 4$, and consequently $\omega' \leq \omega - 3 \times 3/2 + 4 = \omega - 1/2$.

\item When two external strands loop back, we need to prove that $F \leq 3$. The geometry of the graph imposes again $\ell \geq 5$. 
\begin{itemize}
\item If $\ell = 5$, then the external strands must be: $(xbcx)$ (length 2), $(ydey)$ (length 2) and  $(xay)$ (length $1$). Both corners $(bc)$ and $(de)$ are occupied, so that there is no cycle left to support faces of length $2$. Hence $F_2 = 0$ and $F \leq \lfloor (S + F_2 )/3 \rfloor = \lfloor 10/3 \rfloor = 3$. 

\item If $\ell \geq 6$, we may use the constraints $F_2 \leq 2$ and 
$F_3 \leq 2$ imposed by the geometry of the diagram, and conclude that $F \leq \lfloor ( S + 2 F_2 + F_3) /4 \rfloor = \lfloor 15/4 \rfloor = 3$.
\end{itemize}

\end{itemize}

\

\noindent {\bf Graph $H_3$.} There are $4$ vertices so that $S= 21 - \ell$. There is no cycle of length $3$ and $3$ cycles of length $2$, hence $F_3 = 0$ and $F_2 \leq 3$.
\begin{itemize}

\item When all the external strands traverse from $x$ to $y$, $\ell \geq 1+3+3=7$. Therefore
$F \leq \lfloor ( S + 2 F_2 + F_3) /4 \rfloor \leq 5$, which implies $\omega' \leq \omega - 1$.

\item When two external strands loop back, we need to prove that $F \leq 4$. We immediately have $\ell \geq 5$. Given that the external strands that loop back cannot have length $3$, while the strand that traverses cannot 
have length $2$, we furthermore infer $\ell \neq 6$.

\begin{itemize}

\item If $\ell=5$, the external strands must be: $(xay)$, $(xbcx)$ and $(yfgy)$. Two corners $(bc)$ and $(fg)$ being occupied, the only cycle which can support a face of length $2$ is $(de)$. Hence $F_2 \leq 1$, leading to
$F \leq \lfloor ( S + 2 F_2 + F_3) /4 \rfloor \leq 4$.

\item If $\ell=7$, two possible partitions lead to consistent configurations: $7 = 1 + 2 + 4$ and $7 = 3 + 2 + 2$. In the first case, we can assume without loss of generality that the external strands are: $(xbcx)$ (length $2$), $(xay)$ (length $1$) and $(yfdegy)$ (length $4$). The only cycle that can support a face of length $2$ is $(fg)$, hence $F_2 \leq 1$. 
In the second case, two of the external strands are $(xbcx)$ and $(yfgy)$, both of length $2$. This leaves only the cycle $(de)$ for constructing faces of length $2$, hence again $F_2 \leq 1$. We conclude that in both cases $F \leq \lfloor ( S + 2 F_2 + F_3) /4 \rfloor \leq 4$.

\item If $\ell \geq 8$, we immediately obtain $F \leq \lfloor ( S + 2 F_2 + F_3) /4 \rfloor \leq \lfloor (21-8 + 2\times 3) / 4 \rfloor = 4$.

\end{itemize}

\end{itemize}

\noindent {\bf Graph $H_4$.} There are $4$ vertices, so that $S = 21 - \ell$. The geometry of the diagram imposes $F_2 \leq 1$ and $F_3 \leq 3$.

\begin{itemize}
\item When all the external strands traverse from $x$ to $y$, $\ell \geq 1+2+2 = 5$. Hence $F \leq \lfloor (S + F_2 )/3 \rfloor \leq 5$ and $\omega' \leq \omega-1$.

\item When two of the external strands loop back, each of these has length at least $3$, so that $\ell \geq 1+3+3 = 7$. Hence $F \leq \lfloor ( S + 2 F_2 + F_3) /4 \rfloor \leq \lfloor 19/4 \rfloor = 4$, and therefore $\omega' \leq \omega - 1$. 

\end{itemize}

This concludes the proof of Lemma~\ref{lem:2pointmoves}. \hfill $\square$

\section{Wick ordering}

We briefly review in this section the Wick ordering in the usual $\phi^4_2$ quantum field theory. Let us denote $C = \frac{1}{p^2 + m^2}$ the propagator of the model, where $m^2$ is the renormalized (physical) mass 
of the theory. The critical theory is obtained for $m^2=0$. In order to cutoff the UV divergences we consider the $\phi^4_2$ theory on a lattice. 
One first tries to define the $\phi^4_2$ quantum  field theory as the partition function:
\[
 Z = \int [d\phi] \; e^{-\frac{1}{2} \sum_{i,j} \phi_i [C^{-1}]_{ij} \phi_j -\lambda \sum_i \phi_i^4  }  =
 \left[  e^{ \frac{1}{2} \sum_{i,j} \frac{\delta}{\delta \phi_i}  C_{ij}  \frac{\delta}{\delta \phi_j} }  \;e^{-\lambda \sum_i \phi_i^4} \right]_{\phi=0} \;,
\]
where $i = (i_1,i_2) \in  \mathbb{Z}^2$. The covariance $C_{ij}$ is translation invariant, that is it depends only on the Euclidean distance $|i-j|$.
We denote $T = C_{ii}$, which is a constant. 
Unfortunately $Z$ is ill defined as the Feynman graphs can contain tadpoles. The amplitude of a tadpole is $\sum_i \phi_i \phi_i C_{ii} =  T  \sum_i \phi_i \phi_i$ 
and is ultraviolet divergent as $T $ diverges when sending the lattice spacing (UV cutoff) to infinity. 
The theory is renormalized by Wick ordering. The Wick ordered $: \sum_i \phi_i^4 :$ interaction is:
\begin{align*}
 : \sum_i \phi_i^4 : &  = e^{ - \frac{1}{2} \sum_{i,j} \frac{\delta}{\delta \phi_i}  C_{ij}  \frac{\delta}{\delta \phi_j} } \sum_i \phi_i^4  
  = \sum_i \phi_i^4 - 6     \sum_i \phi_i \phi_i C_{ii}   + 3  \sum_i C_{ii}^2  \crcr
 &   =  \sum_i \phi_i^4 - 6 T \sum_i \phi_i \phi_i   + 3 T^2\sum_i 1 \;.
\end{align*}

The renormalized partition function which defines the $\phi^4_2$ theory is:
\begin{align*}
 Z^r & = \int [d\phi] \; e^{-\frac{1}{2} \sum_{i,j} \phi_i [C^{-1}]_{ij} \phi_j -\lambda : \sum_i \phi_i^4 : }  
 = e^{-3\lambda T^2 \sum_i 1 } \int [d\phi] \; e^{-\frac{1}{2} \sum_{i,j} \phi_i[C^{-1} ]_{ij} \phi_j + 6 \lambda T \sum_i \phi_i^2  -\lambda   \sum_i \phi_i^4   } \crcr
  & = e^{-3\lambda T^2 \sum_i 1 } \left[  e^{ \frac{1}{2} \sum_{i,j} \frac{\delta}{\delta \phi_i}  C_{ij}  \frac{\delta}{\delta \phi_j} }  \;e^{ 6 \lambda T \sum_i \phi_i^2-\lambda \sum_i \phi_i^4} \right]_{\phi=0}
  \crcr
  & =  e^{-3\lambda T^2 \sum_i 1 }
  \left[  e^{ \frac{1}{2} \sum_{i,j} \frac{\delta}{\delta \phi_i}  \left( \frac{1}{C^{-1} -  12\lambda T  } \right)_{ij}  \frac{\delta}{\delta \phi_j} }  \;e^{-\lambda \sum_i \phi_i^4} \right]_{\phi=0} \;.
\end{align*}

The renormalized partition function is UV finite in perturbation theory. There are two ways to analyze it.
\begin{description}
 \item[\it The renormalized expansion.] It is obtained by using as covariance the covariance $C$ (with mass the renormalized mass $m^2$). The expansion generates divergent tadpoles and explicit counter terms which subtract the tadpoles to zero: 
    \begin{align*}
      &  
       e^{   e^{ \frac{1}{2} \sum_{i,j} \frac{\delta}{\delta \phi_i}  C_{ij}  \frac{\delta}{\delta \phi_j} }  }  
       \bigg( 6 \lambda T  \sum_i \phi_i^2 -\lambda \sum_i \phi_i^4  \bigg)   \Rightarrow  6 \lambda T \sum_i \phi_i^2 -\lambda \frac{1}{2} 4\cdot 3 \sum_i \phi_i^2 C_{ii} =0 \;.
    \end{align*} 

\item[\it The bare expansion.] It is obtained by using as covariance the bare covariance:
\[ \frac{1}{C^{-1} -  12\lambda T  }   = \frac{1}{ p^2 + m^2 -12\lambda T  } \;,\]
which in particular involves the bare (UV divergent) mass $m^2 -12\lambda T  $. The expansion generates tadpoles. 
The effective two point function, the bare covariance and the self energy are related at the tadpole order by the equations: 
\[
 G^{-1} =  [ C^{-1} -  12\lambda T  ] -\Sigma \;, \qquad \Sigma = - 12 \lambda G_{00} \;,
\]
with solution $G = C$. 
The bare expansion is slightly problematic because the bare mass is negative, hence it seems that the covariance is not positively defined. However, a more careful analysis shows that in a momentum slice the 
$p^2$ term always dominated on the mass, hence it is possible to integrate slice by slice.
\end{description}

In both cases the renormalized partition function has an expansion in Feynman graphs with propagators the physical covariance $C$ (with mass the  renormalized mass $m^2$)
and  having no tadpoles.
In the renormalized expansion the tadpoles are killed one  by one by their counter terms. In the bare expansion the tadpoles are resummed and cancel the mass counterterm to give the renormalized covariance as effective covariance.

 \bibliography{Refs.bib}
 
\end{document}